%% file: main.tex
\documentclass[10pt,twocolumn]{IEEEtran}

\IEEEoverridecommandlockouts                              	 

\usepackage[utf8]{inputenc}
\usepackage{amsmath}
\usepackage{amsthm}
\usepackage{amssymb}
\usepackage{amsfonts}
\usepackage{graphicx}
\usepackage{cite}

\DeclareGraphicsExtensions{.pdf,.jpeg,.png,.es,.jpg,.eps}

\usepackage{caption}
\usepackage{ifthen}
\usepackage{multirow}
\usepackage[multiple]{footmisc}
\newboolean{archive}
\setboolean{archive}{false}

\usepackage{subfigure}
\usepackage{threeparttable}

\usepackage{dsfont}
\usepackage{todonotes}
\usepackage{colortbl}
\usepackage{array}
\usepackage{mathtools}
\usepackage{siunitx}

\input{edits}

\newtheoremstyle{bfnote}%
{}{}%
{\itshape}{}%
{\bfseries}{.}%
{ }%
{\thmname{#1}\thmnumber{ #2}\thmnote{ (#3)}}
\theoremstyle{bfnote}
\newtheorem{thm}{Theorem}
\newtheorem{cor}{Corollary}
\newtheorem{rem}{Remark}
\newtheorem{lem}{Lemma}
\newtheorem{ass}{Assumption}
\newtheorem{dyn-g}{Generator Dynamics}
\newtheorem{dyn-i}{Inverter Dynamics}

\newcommand{\real}[0]{\mathbb R}

\DeclareSymbolFont{bbold}{U}{bbold}{m}{n}
\DeclareSymbolFontAlphabet{\mathbbold}{bbold}
\DeclareMathOperator*{\argmin}{argmin}

\newcommand{\diag}[1]{\ensuremath{\mathrm{diag}\left(#1\right)}}
\newcommand{\tr}[1]{\ensuremath{\mathrm{tr}\left(#1\right)}}

\usepackage{enumitem}
\setlist[enumerate]{leftmargin=*}
\setlist[itemize]{leftmargin=*}

\title{
\huge Dynamic Droop Control in Low-inertia Power Systems
}
\author{Yan Jiang, Richard Pates, and Enrique Mallada%
\thanks{This material is supported by ARO through contract W911NF-17-1-0092, US DoE EERE award DE-EE0008006, NSF through grants CNS 1544771, EPCN 1711188, AMPS 1736448, and CAREER 1752362,
the Swedish Foundation
for Strategic Research, and the Swedish Research Council through the LCCC
Linnaeus Center.
Y. Jiang and E. Mallada are with the Johns Hopkins University, Baltimore, MD 21218, USA. Emails: {\tt \{yjiang,mallada\}@jhu.edu}. Richard Pates is with the Lund University, Box 118, SE-221 00 Lund, Sweden. Email: {\tt{richard.pates@control.lth.se}}. %
A preliminary version of part of the results in this paper has been presented in \cite{m2016cdc} and \cite{y2017cdc}.}  %
}%
\usepackage{setspace}
\begin{document}
\maketitle
\begin{abstract}
A widely embraced approach to mitigate the dynamic degradation in low-inertia power systems
is to mimic generation response using grid-connected inverters to restore the  grid's stiffness.
In this paper, we seek to challenge this approach and advocate for a principled design based on a systematic analysis of the performance trade-offs of inverter-based frequency control.
With this aim, we perform a qualitative and quantitative study comparing the effect of conventional control strategies --droop control (DC) and virtual inertia (VI)-- on several performance metrics induced by $\mathcal L_2$ and $\mathcal L_\infty$ signal norms. By extending a recently proposed modal decomposition method,
we capture the effect of step and stochastic power disturbances, and frequency measurement noise, on the overall transient and steady-state behavior of the system. Our analysis unveils several limitations of these solutions, such as the inability of DC to improve dynamic frequency response without increasing steady-state control effort, or the large frequency variance that VI introduces in the presence of measurement noise.
%
We further propose a novel dynam-i-c Droop controller (iDroop) that overcomes the limitations of DC and VI.
More precisely, we show that iDroop can be tuned to achieve high noise rejection, fast system-wide synchronization, or frequency overshoot (Nadir) elimination without affecting the steady-state control effort share, and propose a tuning recommendation that strikes a balance among these objectives.
Extensive numerical experimentation shows that the proposed tuning is effective even when our proportionality assumptions are not valid, and that the particular tuning used for Nadir elimination strikes a good trade-off among various performance metrics.
\end{abstract}

\section{Introduction}\label{sec:intro}
\input{01-intro.tex}

\section{Preliminaries} \label{sec:prelim}
\input{02-preliminaries.tex}

\section{Results} 
\label{sec:result}
\input{03-results.tex}

\section{The Need for a Better Solution}\label{sec:need}
\input{04-the_need_of_a_better_solution.tex}

\section{Dynam-i-c Droop Control (iDroop)}\label{sec:idroop}
\input{05-idroop.tex}

\section{Numerical Illustrations}\label{sec:simulation}
\input{06-numerical_illustration.tex}

\section{Conclusions}\label{sec:conclusion}
\input{07-conclusion.tex}

\section{Acknowledgements}
The authors would like to acknowledge and thank Fernando Paganini, Petr Vorobev, and Janusz Bialek for their insightful comments that helped improve earlier versions of this manuscript.

\bibliographystyle{IEEEtran}
\bibliography{iDroop}


\end{document}

%% file: edits.tex
\usepackage{ifthen} 
\newboolean{showcomments}
\setboolean{showcomments}{true}   
\usepackage{todonotes}

\definecolor{bleudefrance}{rgb}{0.19, 0.55, 0.91}
\definecolor{ao(english)}{rgb}{0.0, 0.5, 0.0}

\newcommand{\enrique}[1]{  \ifthenelse{\boolean{showcomments}}
{\todo[inline,color=bleudefrance]{Enrique says: #1}}{}}

\newcommand{\addcite}[0]{\ifthenelse{\boolean{showcomments}}
{\textcolor{purple}{(add cite(s)) }}{}}%

\newcommand{\emmargin}[1]{\ifthenelse{\boolean{showcomments}}
{\todo{Enrique: #1)}}{}}

\newboolean{showedits}
\setboolean{showedits}{true}
\usepackage[markup=underlined]{changes}
\definechangesauthor[color=bleudefrance]{EM}
\newcommand{\adde}[1]{
\ifthenelse{\boolean{showedits}}
{\added[id=EM]{#1}}
{\!#1\hspace{-4.75pt}}
}
\newcommand{\repe}[2]{
\ifthenelse{\boolean{showedits}}
{\replaced[id=EM]{#1}{#2}}
{\!#1\hspace{-4.75pt}}
}
\newcommand{\dele}[1]{
\ifthenelse{\boolean{showedits}}
{\deleted[id=EM]{#1}}
{}
}

\definechangesauthor[color=red]{YJ}
\newcommand{\addy}[1]{
\ifthenelse{\boolean{showedits}}
{\added[id=YJ]{#1}}
{\!#1\hspace{-4.75pt}}
}
\newcommand{\repy}[2]{
\ifthenelse{\boolean{showedits}}
{\replaced[id=YJ]{#1}{#2}}
{\!#1\hspace{-4.75pt}}
}
\newcommand{\dely}[1]{
\ifthenelse{\boolean{showedits}}
{\deleted[id=YJ]{#1}}
{}
}

%% file: 01-intro.tex
The shift from conventional synchronous generation to renewable converter-based sources has recently led to a noticeable degradation of the power system frequency dynamics ~\cite{benjamin2017}. At the center of this problem is the reduction of the system-wide inertia that accentuates frequency fluctuations in response to disturbances~\cite{milano2018,ackermann2017}. Besides increasing the risk of frequency instabilities and blackouts~\cite{ulbig2014impact}, this
dynamic degradation also places limits on the total amount
of renewable generation that can be sustained by the grid~\cite{ahmadyar2018}.
Ireland, for instance, is already resorting to wind curtailment
whenever wind becomes larger than $50$\% of existing demand
in order to preserve the grid stability~\cite{o2014studying}.

A widely embraced approach to mitigate this problem is to mimic synchronous generation response using grid-connected converters~\cite{bose2013}. That is, to introduce \emph{virtual inertia} to restore the stiffness that the system used to enjoy~\cite{ofir2018DCvsvVI}. Notable works within this line of research focus on leveraging computational methods~\cite{ Poolla2019place, Guggilam2018TPS, Markovic2019SE} to efficiently allocate synthetic inertial or droop response, or analytical methods that characterize the sensitivity of different performance metrics to global or spatial variations of system parameters~\cite{guo2018cdc,pm2019preprint,pagnier2019optimal}.
However, to this day, it is unclear whether this particular choice of control is the most suitable for the task.
On the one hand, unlike synchronous generators that leverage stored kinetic energy to modulate electric power injection, converter-based controllers need to  actively change their power injection based on noisy measurements of frequency or power.
On the other hand, converter-based control can be significantly faster than conventional generators. Therefore, using converters to mimic generator behavior does not take advantage of their full potential. In this paper, we seek to challenge this approach of mimicking generation response and advocate for a principled control design perspective.

To achieve this goal, we build on recent efforts by the control community on quantifying power network dynamic performance using $\mathcal L_2$ and $\mathcal L_\infty$ norms~\cite{p2017ccc,Poolla2019place}, and perform a systematic  study evaluating the effect of different control strategies, such as droop control (DC)~\cite{Brabandere2007dc} and virtual inertia (VI)~\cite{beck2007virtual}, on a set of static and dynamic figures of merits that are practically relevant from the power engineering standpoint. More precisely, under a mild --yet insightful-- proportionality assumption, we compute closed form solutions and sensitivities of controller parameters on the steady-state control effort share, frequency Nadir, $\mathcal L_2$-synchronization cost, and frequency variance of the response of a power network to step and stochastic disturbances.
Our analysis unveils the inability of DC and VI to cope with seemingly opposing objectives, such as synchronization cost reduction without increasing steady-state effort share (DC), or frequency Nadir reduction without high frequency variance (VI). Therefore, rather than clinging to the idea of efficiently allocating synthetic inertia or droop, we advocate the search of a better solution.

To this end,  we propose novel dynam-i-c Droop (iDroop) control --inspired by classical lead/lag compensation-- which outperforms current control strategies (VI and DC) in an overall sense. More precisely:
\begin{itemize}
    \item Unlike DC that sacrifices steady-state effort share to improve dynamic performance, the added degrees of iDroop allow to decouple steady-state effort  from dynamic performance improvement.
    \item Unlike VI that amplifies frequency measurement noise, the lead/lag property of iDroop makes it less sensitive to noise and power disturbances, as measured by the $\mathcal{H}_2$ norm \cite{g2015tran} of the input-output system defined from measurement noise and power fluctuations to frequency deviations.
    \item iDroop can further be tuned to either eliminate the frequency Nadir, by compensating for the turbine lag, or to eliminate synchronization cost; a feature shown to be unattainable by virtual inertia control.



\end{itemize}
All of above properties are attained through rigorous analysis on explicit expressions for performance metrics that are achieved under a mild yet insightful proportionality assumption that generalizes prior work~\cite{m2016cdc,y2017cdc}.

We further validate our analysis through extensive numerical simulations, performed on a low-inertia system --the Icelandic Grid-- that does not satisfy our parameter assumptions. Our numerical results also show that iDroop with the Nadir eliminated tuning designed based on the proportional parameter assumption works well even in environments with mixed step and stochastic disturbances.

The rest of this paper is organized as follows. Section~\ref{sec:prelim} describes the power network model and defines performance metrics. Section~\ref{sec:result} introduces our assumptions and a system diagonalization that eases the computations and derives some generic results that provide a foundation for further performance analysis. Section~\ref{sec:need} analyzes both  steady-state and dynamic performance of DC and VI, illustrates their limitations, and motivates the need for a new control strategy. Section~\ref{sec:idroop} describes the proposed iDroop and shows how it outperforms DC and VI from different perspectives. Section~\ref{sec:simulation} validates our results through detailed simulations. Section~\ref{sec:conclusion} concludes the paper.

%% file: 02-preliminaries.tex
\subsection{Power System Model}
We consider a connected power network composed of $n$ buses indexed by $i \in \mathcal{V} := \{1,\dots, n\} $ and transmission lines denoted by unordered pairs $\{i,j\} \in \mathcal{E}$, where $\mathcal{E}$ is a set of $2$-element subsets of $\mathcal{V}$. As illustrated by the block diagram in Fig. \ref{fig:model2}, the system dynamics are modeled as a feedback interconnection of bus dynamics and network dynamics.
The input signals $p_\mathrm{in} := \left(p_{\mathrm{in},i}, i \in \mathcal{V} \right) \in \real^n$ and  $d_\mathrm{p} := \left(d_{\mathrm{p},i}, i \in \mathcal{V} \right) \in \real^n$ represent power injection set point changes and power fluctuations around the set point, respectively, and $n_\omega := \left(n_{\omega,i}, i \in \mathcal{V} \right) \in \real^n $ represents frequency measurement noise. The weighting functions $\hat{W}_\mathrm{p}(s)$ and $\hat{W}_\omega{}(s)$ can be used to adjust the size of these disturbances in the usual way. The output signal $\omega:=\left(\omega_i, i \in \mathcal{V} \right) \in \real^n$ represents the bus frequency deviation from its nominal value. We now discuss the dynamic elements in more detail.

\begin{figure}
\centering
\includegraphics[width=\columnwidth]{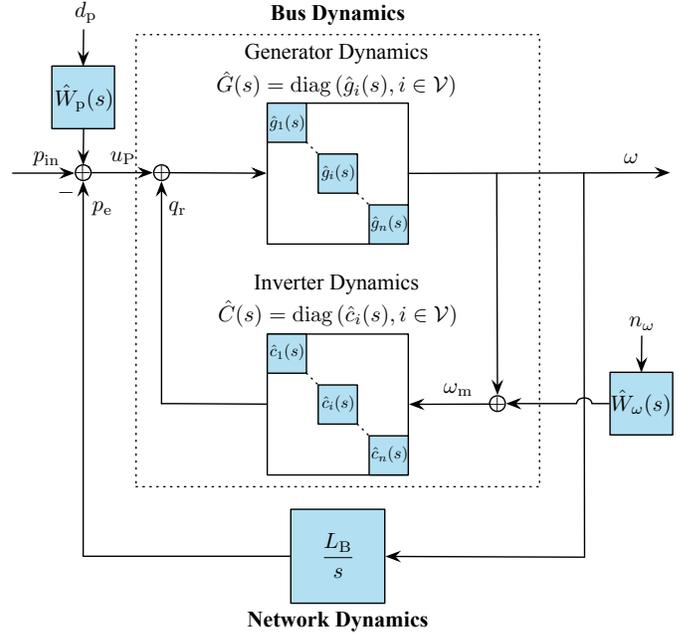}
\caption{Block diagram of power network.}\label{fig:model2}
\end{figure}


\subsubsection{Bus Dynamics} The bus dynamics that maps the net power bus imbalance $u_\mathrm{P} = \left( u_{\mathrm{P},i}, i \in \mathcal{V} \right) \in \real^n$ to the vector of frequency deviations $\omega$ can be described as a feedback loop that comprises a forward-path $\hat{G}(s)$ and a feedback-path $\hat{C}(s)$, where $\hat{G}(s) := \diag {\hat{g}_i(s), i \in \mathcal{V}}$ and $\hat{C}(s) := \diag {\hat{c}_i(s), i \in \mathcal{V}} $ are the transfer function matrices of generators and inverters, respectively.

\paragraph{Generator Dynamics}
The generator dynamics are composed of the standard swing equations with a turbine, i.e.,
\begin{equation}\label{eq:sw}
     	m_i \dot{\omega}_i = - d_i \omega_i + q_{\mathrm{r},i} +q_{\mathrm{t},i} + u_{\mathrm{P},i} \,,
\end{equation}
where $m_i>0$ denotes the aggregate generator inertia, $d_i>0$ the aggregate generator damping, $q_{\mathrm{r},i}$ the controllable input power produced by the grid-connected inverter, and $q_{\mathrm{t},i}$ the change in the mechanical power output of the turbine.
The turbine does not react to the frequency deviation $\omega_i$ until it exceeds a preset threshold $\omega_{\epsilon}\geq0$, i.e.,
\begin{equation}\label{eq:turb}
\tau_i\dot q_{\mathrm{t},i}=\varphi_{\omega_\epsilon}(\omega_i) - q_{\mathrm{t},i}
\end{equation}
with
$$\varphi_{\omega_\epsilon}(\omega_i):=
    \begin{cases}
    -{r_{\mathrm{t},i}^{-1}}(\omega_i+\omega_{\epsilon}) & \omega_i \leq -\omega_{\epsilon}\\
    0 & -\omega_{\epsilon} < \omega_i < \omega_{\epsilon}\\
    -{r_{\mathrm{t},i}^{-1}}(\omega_i-\omega_{\epsilon})  & \omega_i \geq \omega_{\epsilon}
    \end{cases}\,,
$$
where $\tau_i>0$ represents the turbine time constant and $r_{\mathrm{t},i}>0$ the turbine droop coefficient.

Two special cases of our interest are:
\begin{dyn-g}[Standard swing dynamics] \label{dy:sw}
When $|\omega_i(t)| < \omega_{\epsilon}$, the turbines are not triggered and the generator dynamics can be described by the transfer function
\begin{equation}\label{eq:dy-sw}
\hat{g}_i(s) = \frac{1}{m_i s + d_i}
\end{equation}
which is exactly the standard swing dynamics.
\end{dyn-g}

\begin{dyn-g}[Second-order turbine dynamics]
When $\omega_{\epsilon} = 0$, the turbines are constantly triggered and the generator dynamics can be described by the transfer function
\begin{equation} \label{eq:dy-sw-t}
\hat{g}_i(s) = \frac{ \tau_i s + 1 }{m_i \tau_i s^2 + \left(m_i + d_i \tau_i \right) s + d_i + r_{\mathrm{t},i}^{-1}}\;.
\end{equation}
\end{dyn-g}

\paragraph{Inverter Dynamics}
Since power electronics are significantly faster than the electro-mechanical dynamics of generators, we assume that each inverter measures the local grid frequency deviation $\omega_i$ and instantaneously updates the output power $q_{\mathrm{r},i}$. Different control laws can be used to map $\omega_i$ to $q_{\mathrm{r},i}$. We represent such laws using a transfer function $\hat{c}_i(s)$. The two most common ones are:



\begin{dyn-i}[Droop Control]
This control law can provide additional droop capabilities and is given by
\begin{equation} \label{eq:dy-dc}
\hat{c}_i(s) = -r_{\mathrm{r},i}^{-1}\;,
\end{equation}
where $r_{\mathrm{r},i}>0$ is the droop coefficient.
\end{dyn-i}

\begin{dyn-i}[Virtual Inertia]
Besides providing additional droop capabilities, this control law can compensate the loss of inertia and is given by
\begin{equation} \label{eq:dy-vi}
\hat{c}_i(s) = -\left(m_{\mathrm{v},i} s + r_{\mathrm{r},i}^{-1}\right)\;,
\end{equation}
where $m_{\mathrm{v},i}>0$ is the virtual inertia constant.
\end{dyn-i}

\subsubsection{Network Dynamics}
The network power fluctuations $p_\mathrm{e} := \left(p_{\mathrm{e},i}, i \in \mathcal{V} \right) \in \real^n$ are given by a linearized model of the power flow equations~\cite{Purchala2005dc-flow}:
\begin{align}
 \hat p_\mathrm{e}(s) = \frac{L_\mathrm{B}}{s} \hat \omega(s)\;,\label{eq:N}
\end{align}
where $\hat p_\mathrm{e}(s)$ and $\hat \omega(s)$ denote the Laplace transforms of $p_\mathrm{e}$ and $\omega$, respectively.\footnote{We use hat to distinguish the Laplace transform from its time domain counterpart.} The matrix $L_\mathrm{B}$ is an undirected weighted Laplacian matrix of the network with elements
\[
L_{\mathrm{B},{ij}}=\partial_{\theta_j}{\sum_{j=1}^n|V_i||V_j|b_{ij}\sin(\theta_i-\theta_j)}\Bigr|_{\theta=\theta_0}.
\]
Here, $\theta := \left(\theta_i, i \in \mathcal{V} \right) \in \real^n$ denotes the angle deviation from its nominal, $\theta_0 := \left(\theta_{0,i}, i \in \mathcal{V} \right) \in \real^n$ are the equilibrium angles, $|V_i|$ is the (constant) voltage magnitude at bus $i$, and $b_{ij}$ is the line $\{i,j\}$ susceptance.

\subsubsection{Closed-loop Dynamics}
 We will investigate the closed-loop responses of the system in Fig.~\ref{fig:model2} from the power injection set point changes $p_\mathrm{in}$, the power fluctuations around the set point $d_\mathrm{p}$, and frequency measurement noise $n_\omega$ to frequency deviations $\omega$, which can be described compactly by the transfer function matrix
\begin{equation}\label{eq:closed-loop}
    \hat{T}(s) := \begin{bmatrix}\hat{T}_{\omega\mathrm{p}}(s) & \hat{T}_{\omega \mathrm{dn}}(s):=\begin{bmatrix}\hat{T}_{\omega\mathrm{d}} (s) & \hat{T}_{\omega\mathrm{n}} (s)\end{bmatrix}\end{bmatrix}\;.
\end{equation}

\begin{rem}[Model Assumptions]
The \emph{linearized} network model \eqref{eq:closed-loop}
implicitly makes the following assumptions which are standard and well-justified for frequency control on transmission networks \cite{kundur_power_1994}:
\begin{itemize}
\item Bus voltage magnitudes $|V_i|$'s are constant;
we are not modeling the dynamics of exciters used for voltage control; these are assumed to operate at a much faster time-scale.
\item Lines $\{i,j\}$ are lossless.
\item Reactive power flows do not affect bus voltage phase angles and frequencies.
\item Without loss of generality, the equilibrium angle difference ($\theta_{0,i}-\theta_{0,j}$) accross each line is less than $\pi/2$.
\end{itemize}
For a first principle derivation of the model we refer to \cite[Section VII]{Zhao:2013ts}. For applications of similar models for frequency control within the control literature, see, e.g., \cite{Zhao:2014bp,Li:2016tcns,mallada2017optimal}.
\end{rem}

\begin{rem}[Internal Stability of \eqref{eq:closed-loop}]
Throughout this paper we consider feedback interconnections of positive real and strictly positive real subsystems. Internal stability follows from classical results~\cite{khalil2002nonlinear}. Since the focus of this paper is on performance, we do not discuss internal stability here in detail. We refer to the reader to \cite{pm2018tcns}, for a thorough treatment of similar feedback interconnections. From now on a standing assumption --that can be verified-- is that feedback interconnection described in Fig. \ref{fig:model2} is internally stable.
\end{rem}

\subsection{Performance Metrics} \label{ssec:metrics}

Having considered the model of the power network, we are now ready to introduce performance metrics used in this paper to compare different inverter control laws.

\subsubsection{Steady-state Effort Share}
This metric measures the fraction of the power imbalance addressed by inverters, which is calculated as the absolute value of the ratio between the inverter steady-state input power and the total power imbalance, i.e.,
\begin{align}\label{eq:ES}
   \mathrm{ES} := \left|\frac{\sum_{i=1}^n \hat{c}_i(0) \omega_{\mathrm{ss},i}}{\sum_{i=1}^n  p_{\mathrm{in},i}(0^+) }\right|\;,
\end{align}
when the system $\hat{T}_{\omega\mathrm{p}}$ undergoes a step change in power excitation.
Here, $\hat{c}_i(0)$ is the dc gain of the inverter and $\omega_{\mathrm{ss},i}$ is the steady-state frequency deviation.
\subsubsection{Power Fluctuations and Measurement Noise}
This metric measures how the relative intensity of power fluctuations and measurement noise affect the frequency deviations, as quantified by the $\mathcal{H}_2$ norm of the transfer function $\hat{T}_{\omega\mathrm{dn}}$:
\begin{align}
	&\|\hat{T}_{\omega\mathrm{dn}}\|_{\mathcal{H}_2}^2 \label{eq:h2_def_E}\\&:=\!\begin{cases}
	\!\displaystyle\frac{1}{2\pi}\!\!\int_{-\infty{}}^\infty{}\!\!\!\tr{\hat{T}_{\omega\mathrm{dn}}(\boldsymbol{j\omega})^\ast \hat{T}_{\omega\mathrm{dn}}(\boldsymbol{j\omega})}\mathrm{d}\boldsymbol{\omega}&\!\!\textrm{if $\hat{T}_{\omega\mathrm{dn}}$ is stable,}\\
	\!\infty&\!\!\textrm{otherwise.}\nonumber\footnotemark
	\end{cases}
\end{align}
The quantity $\|\hat{T}_{\omega\mathrm{dn}}\|_{\mathcal{H}_2}$ has several standard interpretations in terms of the input-output behavior of the system $\hat{T}_{\omega\mathrm{dn}}$~\cite{g2015tran}. In particular, in the stochastic setting, when the disturbance signals $d_{\mathrm{p},i}$ and $n_{\omega,i}$ are independent, zero mean, unit variance, white noise, then $\lim_{t\to \infty}\mathbb{E} \left[\omega(t)^T \omega(t)\right]=\|\hat{T}_{\omega\mathrm{dn}}\|_{\mathcal{H}_2}^2$.
This means that the sum of the steady-state variances in the output of $\hat{T}_{\omega\mathrm{dn}}$ in response to these disturbance equals the squared $\mathcal{H}_2$ norm of $\hat{T}_{\omega\mathrm{dn}}$. Thus the $\mathcal{H}_2$ norm gives a precise measure of how the intensity of power fluctuations and measurement noise affects the system's frequency deviations.
\footnotetext{$\boldsymbol{j}$ represents the imaginary unit which satisfies $\boldsymbol{j}^2=-1$ and $\boldsymbol{\omega}$ represents the frequency variable.}

\subsubsection{Synchronization Cost}
This metric measures the size of individual bus deviations from the synchronous response when the system $\hat{T}_{\omega\mathrm{p}}$ is subject to a step change in power excitation given by $p_\mathrm{in} = u_0 \mathds{1}_{ t \geq 0 } \in \real^n$, where $u_0 \in \real^n$ is a given vector direction and $\mathds{1}_{ t \geq 0 }$ is the unit-step function \cite{p2017ccc}. This is quantified by the squared $\mathcal{L}_2$ norm of the vector of deviations $\tilde{\omega} := \omega - \bar{\omega} \mathbbold{1}_n \in \real^n$, i.e.,
\begin{equation}\label{eq:sync_cost}
	\|\tilde{\omega}\|_2^2 := \sum_{i=1}^n \int_0^\infty \tilde{\omega}_i(t)^2 \mathrm{d}t\;.
\end{equation}
Here, $\bar{\omega}:= \left(\sum_{i=1}^n m_i\omega_i\right)/\left(\sum_{i=1}^n m_i\right)$ is the system frequency that corresponds to the inertia-weighted average of bus frequency deviations and $\mathbbold{1}_n \in \real^n $ is the vector of all ones.

\subsubsection{Nadir} This metric measures the minimum post-contingency frequency of a power system, which can be quantified by the $\mathcal{L}_{\infty}$ norm of the system frequency $\bar{\omega}$, i.e.,
\begin{equation}\label{eq:Nadir}
	\|\bar{\omega}\|_\infty := \max_{t\geq0} |\bar{\omega}(t)|\;,
\end{equation}
when the system $\hat{T}_{\omega\mathrm{p}}$ has as input a step change in power excitation \cite{p2017ccc}, i.e., $p_\mathrm{in} = u_0 \mathds{1}_{ t \geq 0 } \in \real^n$. This quantity matters in that deeper Nadir increases the risk of under-frequency load shedding and cascading outrages.

%% file: 03-results.tex
In this section we show that under a simplifying assumption, it is possible to compute all of the performance metrics introduced in Section~\ref{ssec:metrics} analytically as functions of the system parameters, which pave us a way to formally compare the conventional control laws DC and VI in Section~\ref{sec:need} as well as suggest an improved control law iDroop in Section~\ref{sec:idroop}. We remark that the assumptions are only used in the analysis, but as we show in Section \ref{sec:simulation} the insights and advantages of the proposed solution are still there when these assumptions do not hold.

\subsection{Diagonalization}

In order to make the analysis tractable, we require the closed-loop transfer functions to be diagonalizable. This is ensured by the following assumption, which is a generalization of \cite{pm2019preprint,p2017ccc}.

\begin{ass}[Proportionality]\label{ass:proportion}
There exists a proportionality matrix $F := \diag {f_i, i \in \mathcal{V}} \in \real_{\geq 0}^{n \times n} $ such that
\[ \hat{G}(s) = \hat{g}_\mathrm{o}(s) F^{-1} \qquad \text{and} \qquad \hat{C}(s) = \hat{c}_\mathrm{o}(s) F\]
where $\hat{g}_\mathrm{o}(s)$ and $\hat{c}_\mathrm{o}(s)$ are called the representative generator and the representative inverter, respectively.
\end{ass}

\begin{rem}[Proportionality parameters]
The parameters $f_i$'s represent the individual machine rating. This definition is rather arbitrary for our analysis, provided that Assumption \ref{ass:proportion} is satisfied. Other alternatives could include $f_i=m_i$ or $f_i=m_i/m$ where $m$ is, for example, either the average or maximum generator inertia. The practical relevance of Assumption~\ref{ass:proportion} is justified, for example, by the empirical values reported in \cite{oakridge2013}, which show that at least in regards of order of magnitude, Assumption \ref{ass:proportion} is a reasonable first-cut approximation to heterogeneity.
\end{rem}

Under Assumption~\ref{ass:proportion},  the representative generator of \eqref{eq:dy-sw} and \eqref{eq:dy-sw-t} are given by
\begin{equation}\label{eq:go-sw}
 \hat{g}_\mathrm{o}(s) = \frac{1}{m s + d}
\end{equation}
and
\begin{equation}\label{eq:go-sw-tb}
 \hat{g}_\mathrm{o}(s) = \frac{ \tau s + 1 }{m \tau s^2 + \left(m + d \tau \right) s + d + r_\mathrm{t}^{-1}}\;, \footnote{We use variables without subscript $i$ to denote parameters of representative generator and inverter.}
\end{equation}
respectively, with $m_i=f_im$, $d_i=f_id$, $r_{\mathrm{t},i}=r_{\mathrm{t}}/f_i$, and $\tau_i=\tau$.

Similarly,
the representative inverters of DC \eqref{eq:dy-dc} and VI \eqref{eq:dy-vi} are given by
\begin{equation}\label{eq:co-dc}
\hat{c}_\mathrm{o}(s) = -r_\mathrm{r}^{-1}
\end{equation}
and
\begin{equation}\label{eq:co-vi}
\hat{c}_\mathrm{o}(s) = -\left(m_\mathrm{v} s + r_\mathrm{r}^{-1}\right)\;,
\end{equation}
with $m_{\mathrm{v},i}=f_im_{\mathrm{v}}$ and $r_{\mathrm{r},i}=r_{\mathrm{r}}/f_i$.


Using Assumption~\ref{ass:proportion}, we can derive a diagonalized version of \eqref{eq:closed-loop}. First, we rewrite
\[\hat{G}(s) = F^{-\frac{1}{2}} [\hat{g}_\mathrm{o}(s)I_n] F^{-\frac{1}{2}} \quad\text{and}\quad \hat{C}(s) = F^{\frac{1}{2}} [\hat{c}_\mathrm{o}(s)I_n] F^{\frac{1}{2}}\]
as shown in Fig. \ref{fig:diag1}, and after a loop transformation obtain Fig. \ref{fig:diag2}. Then, we define the scaled Laplacian matrix
\begin{equation}\label{eq:scale-L}
L_\mathrm{F} := F^{-\frac{1}{2}} L_\mathrm{B} F^{-\frac{1}{2}}
\end{equation}
by grouping the terms in the upper block of Fig. \ref{fig:diag2}. Moreover, since $L_\mathrm{F} \in \real^{n \times n}$ is symmetric positive semidefinite, it is real orthogonally diagonalizable with non-negative eigenvalues~\cite{Horn2012MA}. Thus, there exists an orthogonal matrix $V \in \real^{n \times n}$ with $V^T V = V V^T = I_n$, such that
\begin{equation}\label{eq:ortho-diag}
L_\mathrm{F} = V \Lambda V^T\;,
\end{equation}
where $\Lambda := \diag{\lambda_k, k \in \{1,\dots, n\}} \in \real_{\geq0}^{n \times n}$ with $\lambda_k$ being the $k$th eigenvalue of $L_\mathrm{F}$ ordered non-decreasingly $(0 = \lambda_1 < \lambda_2 \leq \ldots \leq\lambda_n)$\footnote{Recall that we assume the power network is connected, which means that $L_\mathrm{F}$ has a single eigenvalue at the origin.} and $V := \begin{bmatrix} (\sum_{i=1}^n f_i)^{-\frac{1}{2}} F^{\frac{1}{2}} \mathbbold{1}_n & V_{\bot} \end{bmatrix}$ with $V_{\bot} := \begin{bmatrix} v_2 &\ldots & v_n \end{bmatrix}$ composed by the eigenvector $v_k$ associated with $\lambda_k$.\footnote{We use $k$ and $l$ to index dynamic modes but $i$ and $j$ to index bus numbers.} Now, applying \eqref{eq:scale-L} and \eqref{eq:ortho-diag} to Fig. \ref{fig:diag2} and rearranging blocks of $V$ and $V^T$ results in Fig. \ref{fig:diag3}. Finally, moving the block of $\hat{c}_\mathrm{o}(s) I_n$ ahead of the summing junction and combining the two parallel paths produces Fig. \ref{fig:block-diag}, where the boxed part is fully diagonalized.

Now, by defining
the closed-loop with a forward-path $\hat{g}_\mathrm{o}(s) I_n$ and a feedback-path $\left(\Lambda/s - \hat{c}_\mathrm{o}(s) I_n\right)$ as
\begin{equation*}
\hat{H}_\mathrm{p}(s) = \diag{\hat{h}_{\mathrm{p},k}(s), k \in \{1,\dots, n\}}
\end{equation*}
where
\begin{equation} \label{eq:hp-s}
\hat{h}_{\mathrm{p},k}(s) = \frac{\hat{g}_\mathrm{o}(s)}{1+\hat{g}_\mathrm{o}(s)\left(\lambda_k/s-\hat{c}_\mathrm{o}(s)\right)}\;,
\end{equation}
and $\hat{H}_\omega(s) = \hat{c}_\mathrm{o}(s) \hat{H}_\mathrm{p}(s)$, i.e.,
\begin{equation*}
\hat{H}_\omega(s) = \diag{\hat{h}_{\omega,k}(s), k \in \{1,\dots, n\}}
\end{equation*}
where
\begin{equation} \label{eq:homega-s}
\hat{h}_{\omega,k}(s) = \hat{c}_\mathrm{o}(s)\hat{h}_{\mathrm{p},k}(s)\;,
\end{equation}
the closed-loop transfer functions from $p_\mathrm{in}$, $d_\mathrm{p}$, and $n_\omega$ to $\omega$ become
\begin{subequations}\label{eq:T-diag}
\begin{equation}\label{eq:Tp}
\hat{T}_{\omega\mathrm{p}} (s) = F^{-\frac{1}{2}} V \hat{H}_\mathrm{p}(s) V^T F^{-\frac{1}{2}}\;,
\end{equation}
\begin{equation}\label{eq:Td}
\hat{T}_{\omega\mathrm{d}} (s) = F^{-\frac{1}{2}} V \hat{H}_\mathrm{p}(s) V^T F^{-\frac{1}{2}}\hat{W}_\mathrm{p}(s)\;,
\end{equation}
\begin{equation}
\hat{T}_{\omega\mathrm{n}} (s) = F^{-\frac{1}{2}} V \hat{H}_\omega(s) V^T F^{\frac{1}{2}}\hat{W}_\omega{}(s)\;,
\end{equation}
\end{subequations}
respectively.

Note that depending on the specific generator and inverter dynamics involved, we may add subscripts in the name of a transfer function without making a further declaration in the rest of this paper. For example, we may add 'T' if the turbine is triggered and 'DC' if the inverter operates in DC mode as in $\hat{h}_{\mathrm{p},k,\mathrm{T,DC}}(s)$.

\begin{figure*}[!t]
\centering
\subfigure[]
{\includegraphics[width=0.32\textwidth]{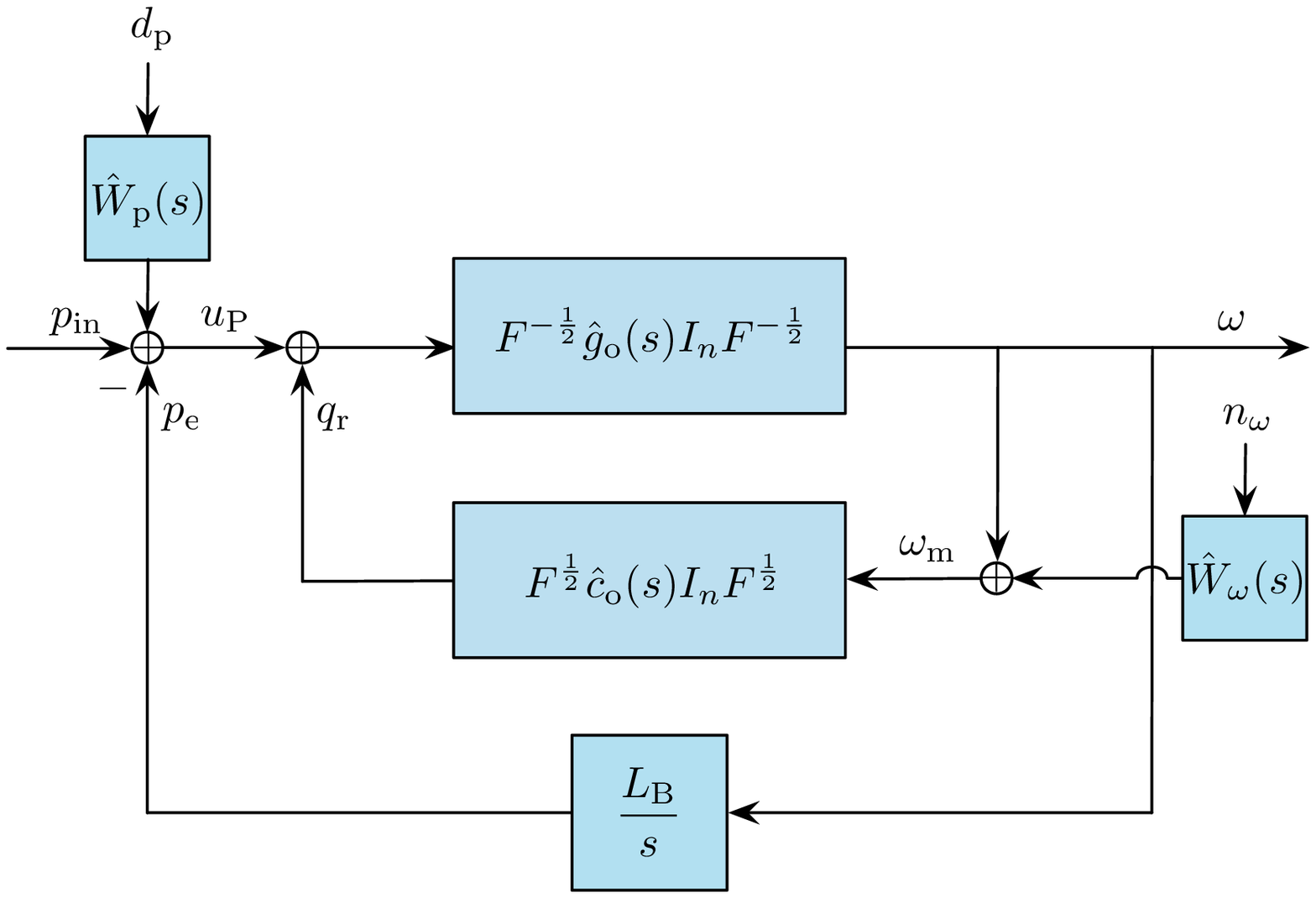}\label{fig:diag1}}
\hfil
\subfigure[]
{\includegraphics[width=0.32\textwidth]{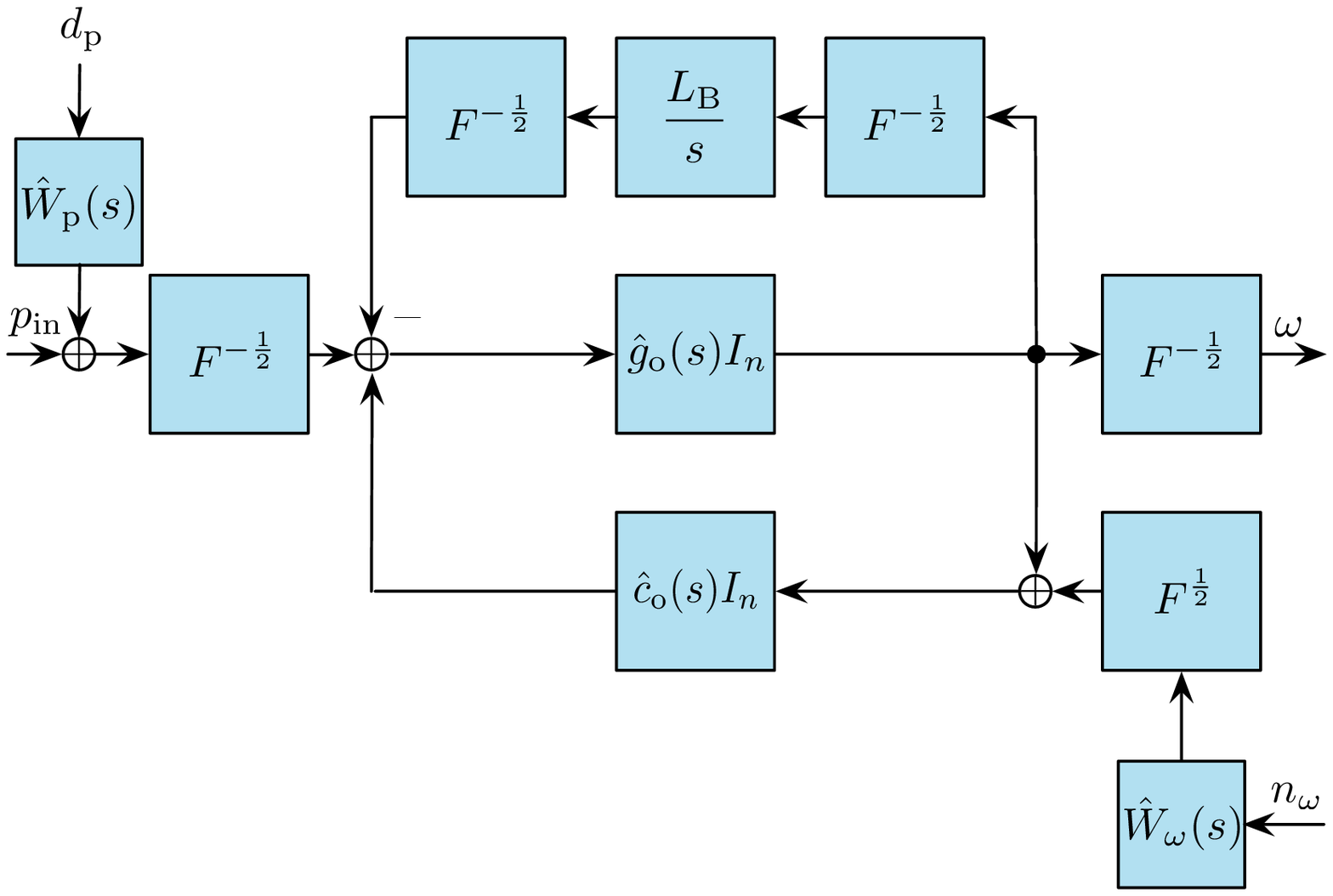}\label{fig:diag2}}
\hfil
\subfigure[]
{\includegraphics[width=0.32\textwidth]{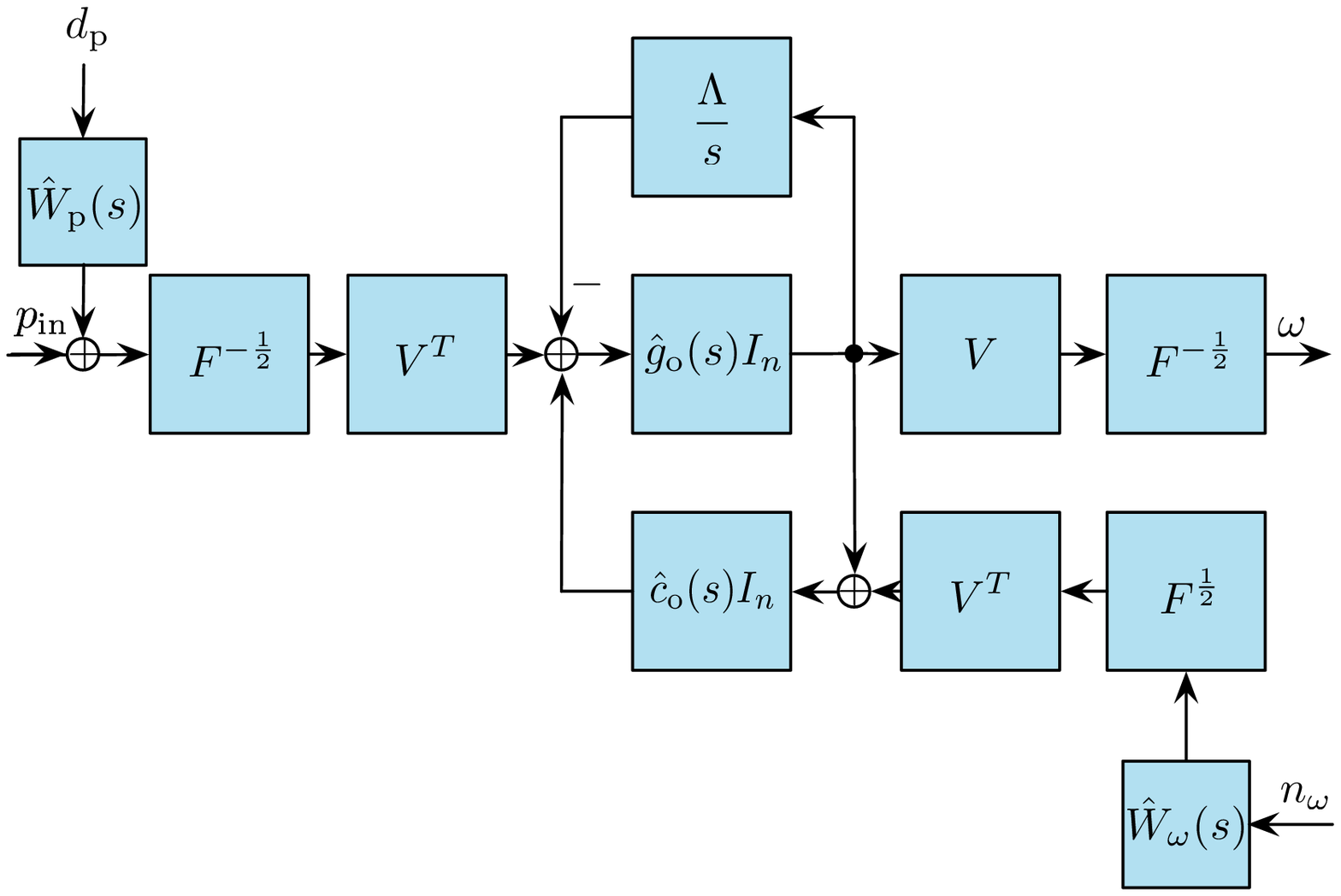}\label{fig:diag3}}
\caption{Equivalent block diagrams of power network under proportionality assumption.}
\label{fig:block-diag-process}
\end{figure*}

\begin{figure}
\centering
\includegraphics[width=\columnwidth]{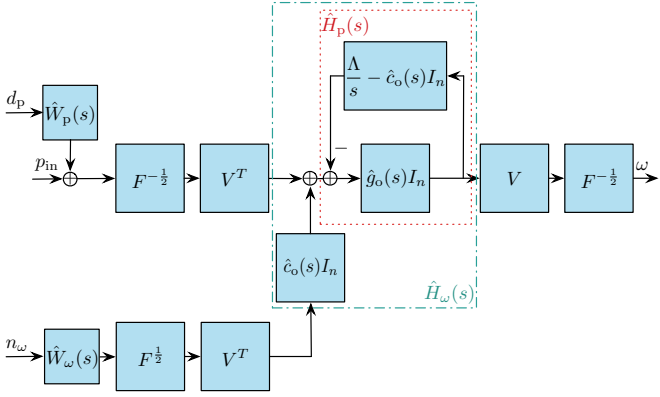}
\caption{Diagonalized block diagram of power network.}\label{fig:block-diag}
\end{figure}



\subsection{Generic Results for Performance Metrics}
We now derive some important building blocks required for the performance analysis of the system $\hat{T}$ described in \eqref{eq:T-diag}. As described in Section \ref{ssec:metrics}, the sensitivity to power fluctuations and measurement noise can be evaluated through the $\mathcal{H}_2$ norm of the system $\hat{T}_{\omega\mathrm{dn}}$, while the steady-state effort share, synchronization cost, and Nadir can all be characterized by a step response of the system $\hat{T}_{\omega\mathrm{p}}$. There are two scenarios that are of our interest.

\begin{ass}[Proportional weighting scenario]\label{ass:noise}\
\begin{itemize}
\item
The noise weighting functions are given by
\begin{equation*}
    \hat{W}_\mathrm{p}(s) = \kappa_\mathrm{p} F^{\frac{1}{2}} \qquad \text{and} \qquad \hat{W}_\omega(s) = \kappa_\omega F^{-\frac{1}{2}},
\end{equation*}
where $\kappa_\mathrm{p}>0$ and $\kappa_\omega>0$ are weighting constants.
\item
$|\omega_i(t)| < \omega_{\epsilon}$, $\forall i \in \mathcal{V}$ and $t\geq0$ such that turbines will not be triggered.
\end{itemize}
\end{ass}
\begin{ass}[Step input scenario]\label{ass:step}\
\begin{itemize}
\item
There is a step change as defined in Section \ref{ssec:metrics} on the power injection set point, i.e., $p_\mathrm{in} = u_0 \mathds{1}_{ t \geq 0 }$, $d_\mathrm{p}= \mathbbold{0}_n$, and $n_\omega = \mathbbold{0}_n$ with $\mathbbold{0}_n \in \real^n $ being the vector of all zeros.
\item
$\omega_{\epsilon} = 0$ such that turbines are constantly triggered.
\end{itemize}
\end{ass}
\begin{rem}[Weighting assumption]
As a natural counterpart of Assumption~\ref{ass:proportion}, we look at the case when the power fluctuations and measurement noise are weighted directly and inversely proportional to the square root of the bus ratings, respectively. In the case of $\hat W_\mathrm{p}(s)$, this is equivalent to assuming that demand fluctuation variances are proportional to the bus ratings, which is in agreement with the central limit theorem. For $\hat W_\mathrm{\omega}(s)$, this is equivalent to assuming the frequency measurement noise variances are inversely proportional to the bus ratings, which is in line with the inverse relationship between jitter variance and power consumption for an oscillator in phase-locked-loop \cite{Weigandt1994}.
\end{rem}

\subsubsection{Steady-state Effort Share}
As indicated by \eqref{eq:ES}, the key of computing the steady-state effort share lies in computing the steady-state frequency deviation $\omega_{\mathrm{ss}}$ of the system $\hat{T}_{\omega\mathrm{p}}$. When the system synchronizes, the steady-state frequency deviation is given by $\omega_{\mathrm{ss}} = \omega_{\mathrm{syn}} \mathbbold{1}_n$
%
and $\omega_{\mathrm{syn}}$ is called the synchronous frequency. In the absence of a secondary control layer, e.g., automatic generation control \cite{d1973tran}, the system can synchronize with a nontrivial frequency deviation, i.e., $\omega_{\mathrm{syn}} \neq 0$.

The following lemma provides a general expression for $\omega_{\mathrm{syn}}$ in our setting.

\begin{lem}[Synchronous frequency]\label{lem:syn-fre}
Let Assumption~\ref{ass:step} hold.
If $q_{\mathrm{r},i}$ is determined by a control law $\hat{c}_i(s)$, then the output $\omega$ of the system $\hat{T}_{\omega\mathrm{p}}$ synchronizes to the steady-state frequency deviation $\omega_{\mathrm{ss}} = \omega_{\mathrm{syn}} \mathbbold{1}_n$ with
\begin{equation}
 \omega_{\mathrm{syn}} = \dfrac{\sum_{i=1}^n  u_{0,i}}{\sum_{i=1}^n \left( d_i + {r_{\mathrm{t},i}^{-1} - \hat{c}_i(0)} \right)}\;. \label{eq:ome-syn}
\end{equation}
\end{lem}

\begin{proof}
Combining \eqref{eq:sw} and \eqref{eq:N} through the relationship $u_\mathrm{P} = p_\mathrm{in} - p_\mathrm{e}$,
we get the (partial) state-space representation of the system $\hat{T}_{\omega\mathrm{p}}$ as
\begin{subequations}\label{eq:ss}
	\begin{align}
		\dot{\theta} =&\ \omega \,,\\
     	M \dot{\omega} =& -D \omega -L_\mathrm{B} \theta + q_\mathrm{r} + q_\mathrm{t} + p_\mathrm{in} \,, \label{eq:ss-fre}
	\end{align}
\end{subequations}
where $M := \diag{m_i, i \in \mathcal{V}} \in \real_{\geq0}^{n \times n}$, $D := \diag{d_i, i \in \mathcal{V}} \in \real_{\geq0}^{n \times n}$, $q_\mathrm{r} := \left(q_{\mathrm{r},i}, i \in \mathcal{V} \right) \in \real^n$, and $q_\mathrm{t} := \left(q_{\mathrm{t},i}, i \in \mathcal{V} \right) \in \real^n$.
In steady-state, \eqref{eq:ss} yields
\begin{equation}\label{eq:ss-pf1}
    L_\mathrm{B} \omega_{\mathrm{ss}} t = -D \omega_{\mathrm{ss}} -L_\mathrm{B} \theta_{\mathrm{ss}_0} + q_{\mathrm{r},\mathrm{ss}} + q_{\mathrm{t},\mathrm{ss}} + u_0 \,,
\end{equation}
where $(\theta_{\mathrm{ss}_0} + \omega_{\mathrm{ss}} t, \omega_{\mathrm{ss}}, q_{\mathrm{r},\mathrm{ss}}, q_{\mathrm{t},\mathrm{ss}})$ denotes the steady-state solution of \eqref{eq:ss}. Equation \eqref{eq:ss-pf1} indicates that $L_\mathrm{B} \omega_{\mathrm{ss}} t$ is constant and thus $L_\mathrm{B} \omega_{\mathrm{ss}} = \mathbbold{0}_n$. It follows that $\omega_{\mathrm{ss}} = \omega_{\mathrm{syn}} \mathbbold{1}_n$.
Therefore, \eqref{eq:ss-pf1} becomes
\begin{align}
\mathbbold{0}_n =& -D \omega_{\mathrm{syn}} \mathbbold{1}_n -L_\mathrm{B} \theta_{\mathrm{ss}_0} + q_{\mathrm{r},\mathrm{ss}} + q_{\mathrm{t},\mathrm{ss}} + u_0\;, \label{eq:synsw1}
\end{align}
where $q_{\mathrm{r},\mathrm{ss}} = \left( \hat{c}_i(0) \omega_{\mathrm{syn}}, i \in \mathcal{V} \right)\in \real^n$ and  $q_{\mathrm{t},\mathrm{ss}} = \left(-r_{\mathrm{t},i}^{-1} \omega_{\mathrm{syn}}, i \in \mathcal{V} \right)\in \real^n$ when $\omega_{\epsilon}=0$ by \eqref{eq:turb}. Pre-multiplying \eqref{eq:synsw1} by $\mathbbold{1}_n^T$ and using the property that $\mathbbold{1}_n^T L_\mathrm{B} = \mathbbold{0}_n^T$, we get the desired result in \eqref{eq:ome-syn}.
\end{proof}

Now, the theorem below provides an explicit expression for the steady-state effort share.

\begin{thm}[Steady-state effort share]\label{thm:ss-es}
Let Assumption~\ref{ass:step} hold. If $q_{\mathrm{r},i}$ is determined by a control law $\hat{c}_i(s)$, then the steady-state effort share of the system {$\hat{T}_{\omega\mathrm{p}}$} is given by
\begin{equation}
    \mathrm{ES} = \left|\frac{\sum_{i=1}^n \hat{c}_i(0)}{\sum_{i=1}^n  \left( d_i + {r_{\mathrm{t},i}^{-1} - \hat{c}_i(0)} \right) }\right|\;.
\end{equation}
\end{thm}
\begin{proof}
It follows directly from Lemma~\ref{lem:syn-fre} that $\omega_{\mathrm{ss},i}=\omega_{\mathrm{syn}}$ and $\sum_{i=1}^n  u_{0,i}=\omega_{\mathrm{syn}} \sum_{i=1}^n \left( d_i + {r_{\mathrm{t},i}^{-1} - \hat{c}_i(0)} \right)$. Plugging these two equations to the definition of ES in \eqref{eq:ES} yields the desired result.
\end{proof}

\subsubsection{Power Fluctuations and Measurement Noise}
We seek to characterize the effect of power fluctuations and frequency measurement noise on the frequency variance, i.e., the $\mathcal{H}_2$ norm of the system $\hat{T}_{\omega\mathrm{dn}}$.

We first show that the squared $\mathcal{H}_2$ norm of $\hat{T}_{\omega\mathrm{dn}}$ is a weighted sum of the squared $\mathcal{H}_2$ norm of each $\hat{h}_{\mathrm{p},k}$ and $\hat{h}_{\omega,k}$ in the diagonalized system \eqref{eq:T-diag}.

\begin{thm}[Frequency variance]\label{thm:h2-sum}
Define $\Gamma := V^T F^{-1} V$. If Assumptions~\ref{ass:proportion} and \ref{ass:noise} hold, then
 \begin{align*}
    \|\hat{T}_{\omega\mathrm{dn} }\|_{\mathcal{H}_2}^2 = \sum_{k=1}^n\Gamma_{kk}\left(\kappa_\mathrm{p}^2 \|\hat{h}_{\mathrm{p},k}\|_{\mathcal{H}_2}^2 + \kappa_\omega^2 \|\hat{h}_{\omega,k}\|_{\mathcal{H}_2}^2\right)\;.
\end{align*}
\end{thm}

\begin{proof}

It follows from \eqref{eq:closed-loop} and \eqref{eq:h2_def_E} that
\begin{align*}
    \|\hat{T}_{\omega\mathrm{dn} }\|_{\mathcal{H}_2}^2 \!=&\ \frac{1}{2\pi}\int_{-\infty{}}^\infty \tr{\hat{T}_{\omega\mathrm{d}} (\boldsymbol{j\omega})^\ast \hat{T}_{\omega\mathrm{d}} (\boldsymbol{j\omega}) }\,\mathrm{d}\boldsymbol{\omega} \nonumber\\&+ \frac{1}{2\pi}\int_{-\infty{}}^\infty \tr{ \hat{T}_{ \omega\mathrm{n}} (\boldsymbol{j\omega})^\ast \hat{T}_{ \omega\mathrm{n}} (\boldsymbol{j\omega})}\,\mathrm{d}\boldsymbol{\omega}\\=:&\ \|\hat{T}_{ \omega\mathrm{d}}\|^2_{\mathcal{H}_2}+\|\hat{T}_{ \omega\mathrm{n}}\|^2_{\mathcal{H}_2}.
\end{align*}
We now compute $\|\hat{T}_{ \omega\mathrm{d}}\|^2_{\mathcal{H}_2}$. Using \eqref{eq:Td} and the fact that $\hat{W}_\mathrm{p} (s) = \kappa_\mathrm{p} F^{\frac{1}{2}}$ by Assumption \ref{ass:noise}, we get $\hat{T}_{ \omega\mathrm{d}}(s)=\kappa{}_\mathrm{p}F^{-\frac{1}{2}} V \hat{H}_\mathrm{p}(s) V^T$.
Therefore,
\begin{equation*}
\hat{T}_{\omega\mathrm{d}} (\boldsymbol{j\omega})^\ast \hat{T}_{\omega\mathrm{d}} (\boldsymbol{j\omega}) =\kappa{}_\mathrm{p}^2 V \hat{H}_\mathrm{p}(\boldsymbol{j\omega})^\ast V^T F^{-1} V \hat{H}_\mathrm{p}(\boldsymbol{j\omega}) V^T.
\end{equation*}
Using the cyclic property of the trace, this implies that
\begin{equation*}
\tr{\hat{T}_{\omega\mathrm{d}} (\boldsymbol{j\omega})^\ast \hat{T}_{\omega\mathrm{d}} (\boldsymbol{j\omega}) }=\kappa{}_\mathrm{p}^2 \tr{\hat{H}_\mathrm{p}(\boldsymbol{j\omega})^\ast \Gamma{} \hat{H}_\mathrm{p}(\boldsymbol{j\omega}) },
\end{equation*}
where $\Gamma:=V^TF^{-1}V$. Therefore, it follows that
\begin{align*}
    \|\hat{T}_{\omega\mathrm{d}}\|^2_{\mathcal{H}_2}&=\frac{1}{2\pi}\int_{-\infty{}}^\infty{}\kappa{}_\mathrm{p}^2 \tr{\hat{H}_\mathrm{p}(\boldsymbol{j\omega})^* \Gamma{} \hat{H}_\mathrm{p}(\boldsymbol{j\omega}) }\,\mathrm{d}\boldsymbol{\omega}\\
    = \sum_{k=1}^n \frac{\kappa_\mathrm{p}^2\Gamma_{kk}}{2\pi}&\int_{-\infty}^\infty \left|\hat{h}_{\mathrm{p},k}(\boldsymbol{j\omega})\right|^2\, \mathrm{d}\boldsymbol{\omega} = \kappa_\mathrm{p}^2 \sum_{k=1}^n \Gamma_{kk}\|\hat{h}_{\mathrm{p},k}\|_{\mathcal{H}_2}^2\;.
\end{align*}
The result follows from a similar argument on $\|\hat{T}_{\omega\mathrm{n}}\|^2_{\mathcal{H}_2}$.
\end{proof}
Theorem \ref{thm:h2-sum} allows us to compute the $\mathcal{H}_2$ norm of $\hat{T}_{\omega\mathrm{dn}}$ by means of computing the norms of a set of simple scalar transfer functions.
However, for different controllers, the transfer functions $\hat{h}_{\mathrm{p},k}$ and $\hat{h}_{\omega,k}$ will change. Since in all the cases these transfer functions are of fourth-order or lower, the following lemma will suffice for the purpose of our comparison.

\begin{lem}[$\mathcal{H}_2$ norm of a fourth-order transfer function]\label{lm:h2-4th}
Let
\[
\hat{h}(s)=\frac{b_3s^3+b_2s^2+b_1s+b_0}{s^4+a_3s^3+a_2s^2+a_1s+a_0}+b_4
\]
be a stable transfer function.
If $b_4=0$, then
\begin{equation}\label{eq:4throderh2}
\|\hat{h}\|_{\mathcal{H}_2}^2 =
\displaystyle{\frac{\zeta_0 b_0^2+\zeta_1 b_1^2+\zeta_2 b_2^2+\zeta_3 b_3^2+\zeta_4}{2 a_0 \left(a_1 a_2 a_3 -a_1^2 -a_0 a_3^2\right)}}\;,
\end{equation}
where
\begin{align}\label{eq:zeta}
    \zeta_0:=&\ a_2 a_3-a_1 \,,\qquad \zeta_1:=\ a_0 a_3\,,\qquad \zeta_2:=\ a_0a_1\,,\\
    \zeta_3:=&\ a_0a_1 a_2 - a_0^2 a_3\,,\qquad
    \zeta_4:=-2a_0(a_1 b_1 b_3  +a_3 b_0 b_2)\,.\nonumber
\end{align}
Otherwise, $\|\hat{h}\|_{\mathcal{H}_2}^2 = \infty$.
\end{lem}

\begin{proof}
First recall that given any state-space realization of $\hat{h}(s)$, the $\mathcal{H}_2$ norm can be calculated by solving a particular Lyapunov equation. More specifically, suppose
\[
\Sigma_{\hat{h}(s)}=\left[\begin{array}{c|c}
    A & B \\\hline{}
    C & D
\end{array}\right],
\]
and let $X$ denote the solution to the Lyapunov equation
\begin{align}
AX+XA^T=-BB^T.\label{eq:lyp-x}
\end{align}
If $\hat{h}(s)$ is stable, then
\begin{align}
\|\hat{h}\|_{\mathcal{H}_2}^2=\begin{cases}
\infty{}&\text{if $D\neq{}0$,}\\
CXC^T&\text{otherwise}.
\end{cases}\label{eq:h2-2cases}
\end{align}

Consider the observable canonical form of $\hat{h}(s)$ given by
\begin{align}
\Sigma_{\hat{h}(s)}=
\left[\begin{array}{cccc|c}
    0&0&0&-a_0&b_0\\
    1&0&0&-a_1&b_1\\
    0&1&0&-a_2&b_2\\
    0&0&1&-a_3&b_3\\\hline
    0&0&0 & 1&b_4
\end{array}\right].\label{eq:real-h}
\end{align}
Since $D=b_4$, it is trivial to see from \eqref{eq:h2-2cases} that if $b_4\neq0$ then $\|\hat{h}\|_{\mathcal{H}_2}^2 = \infty$. Hence, in the rest of the proof, we assume $b_4=0$. We will now solve the Lyapunov equation analytically for the realization \eqref{eq:real-h}. $X$ must be symmetric and thus can be parameterized as
\begin{equation} \label{eq:grammian-4th}
X=\big[x_{ij}\big]\in\real^{4\times4}\;, \quad\text{with}\quad x_{ij}=x_{ji}.
\end{equation}
Since it is easy to see that $C X C^T = x_{44}$, the problem becomes solving for $x_{44}$.
Substituting \eqref{eq:real-h} and \eqref{eq:grammian-4th} into \eqref{eq:lyp-x} yields the following equations
\begin{subequations}\label{eq:lyap-group}
    \begin{align}
    2 a_0 x_{14}=&\ b_0^2\;,\label{eq:lyap-1}\\
    x_{12} - a_2 x_{14} - a_0 x_{34} =& -b_0 b_2\;,\label{eq:lyap-3}\\
     2(x_{12} - a_1 x_{24})=&-b_1^2\;,\label{eq:lyap-5}\\
     x_{23} - a_3 x_{24} + x_{14} - a_1 x_{44} =& -b_1 b_3\;,\label{eq:lyap-7}\\
     2(x_{23} - a_2 x_{34})=&-b_2^2\;,\label{eq:lyap-8}\\
     2(x_{34} - a_3 x_{44})=&-b_3^2\;.\label{eq:lyap-10}
    \end{align}
\end{subequations}
\ifthenelse{\boolean{archive}}{
Since $\hat{h}(s)$ is stable, by the Routh-Hurwitz criterion $a_0\neq0$, and therefore \eqref{eq:lyap-1} yields
\begin{align}
    x_{14}=\frac{b_0^2}{2 a_0}\;.\label{eq:x14}
\end{align}
Applying \eqref{eq:x14} to \eqref{eq:lyap-3} and \eqref{eq:lyap-7} gives
\begin{subequations}
\begin{align}
    x_{12}  =& a_0 x_{34} +  \frac{a_2b_0^2}{2 a_0}-b_0 b_2\;,\label{eq:x12 in x34}\\
    x_{23} - a_3 x_{24}   =& a_1 x_{44}-\frac{b_0^2}{2 a_0}-b_1 b_3\;.\label{eq:x23-24 in x44}
\end{align}
\end{subequations}
We now parameterize unknowns in $x_{44}$.
Equation \eqref{eq:lyap-10} yields
\begin{align}
    x_{34} =& a_3 x_{44}-\frac{b_3^2}{2}\;.\label{eq:x34 in x44}
\end{align}
Substituting \eqref{eq:x34 in x44} into \eqref{eq:lyap-8} and \eqref{eq:x12 in x34} gives
\begin{subequations}
\begin{align}
    x_{23} =&a_2 a_3 x_{44}-\frac{a_2 b_3^2+b_2^2}{2}\;,\label{eq:x23 in x44}\\
    x_{12}  =& a_0 a_3 x_{44}-\frac{a_0b_3^2}{2} +  \frac{a_2b_0^2}{2 a_0}-b_0 b_2\;,\label{eq:x12 in x44}
\end{align}
\end{subequations}
respectively.
Plugging \eqref{eq:x12 in x44} into \eqref{eq:lyap-5} leads to
\begin{align}
      a_1 x_{24}=&a_0 a_3 x_{44}-\frac{a_0b_3^2}{2} +  \frac{a_2b_0^2}{2 a_0}-b_0 b_2 +\frac{b_1^2}{2}\;,\label{eq:x24 in x44}
\end{align}
Combining \eqref{eq:x23-24 in x44}, \eqref{eq:x23 in x44}, and \eqref{eq:x24 in x44}, we can solve for $x_{44}$ as}
{Through standard algebra, we can solve for $x_{44}$ as}
\[
x_{44} =
\displaystyle{\frac{\zeta_0 b_0^2+\zeta_1 b_1^2+\zeta_2 b_2^2+\zeta_3 b_3^2+\zeta_4}{2 a_0 \left(a_1 a_2 a_3-a_1^2 -a_0 a_3^2\right)}}
\]
with $\zeta_0, \zeta_1, \zeta_2, \zeta_3$, and $\zeta_4$ defined by \eqref{eq:zeta},
which concludes the proof; the denominator is guaranteed to be nonzero by the Routh-Hurwitz criterion.
\end{proof}

\begin{rem}[$\mathcal{H}_2$ norm of a transfer function lower than fourth-order]\label{rem:h2-3rd}
Although Lemma~\ref{lm:h2-4th} is stated for a fourth-order transfer function, it can also be used to find the $\mathcal{H}_2$ norm of third-, second-, and first-order transfer functions by considering appropriate limits. For example, setting $a_0=b_0=\epsilon{}$ and considering the limit $\epsilon\to 0$, \eqref{eq:4throderh2} gives the $\mathcal{H}_2$ norm of a generic third-order transfer function. This process shows that given a stable transfer function $\hat{h}(s)$, if $b_4=0$ and:
\begin{itemize}
    \item (third-order transfer function) $a_0=b_0=0$, then
    \[
    \|\hat{h}\|_{\mathcal{H}_2}^2 = \frac{a_3 b_1^2+a_1 b_2^2+a_1 a_2b_3^2-2 a_1 b_1 b_3}{2 a_1 (a_2 a_3- a_1)};
    \]
    \item (second-order transfer function) $a_0=b_0=a_1=b_1=0$, then
    \[
    \|\hat{h}\|_{\mathcal{H}_2}^2 = \frac{b_2^2+a_2b_3^2}{2 a_2a_3};
    \]
    \item (first-order transfer function) $a_0=b_0=a_1=b_1=a_2=b_2=0$, then
    \[
    \|\hat{h}\|_{\mathcal{H}_2}^2 =\frac{b_3^2}{2 a_3};
    \]
\end{itemize}
otherwise $\|\hat{h}\|_{\mathcal{H}_2}^2=\infty{}$.
\end{rem}

\begin{rem}[Well-definedness by the stability]
Note that the stability of $\hat{h}(s)$ guarantees that the denominators in all the above $\mathcal{H}_2$ norm expressions are nonzero by the Routh-Hurwitz stability criterion.
\end{rem}

\subsubsection{Synchronization Cost}
The computation of the synchronization cost defined in \eqref{eq:sync_cost} for the system $\hat{T}_{ \omega\mathrm{p}}$ in the absence of inverter control can be found in \cite{pm2019preprint}. Taking this into account, we can get corresponding results for the system with any control law readily.
\begin{lem}[Synchronization cost]\label{lem:syncost-generic}
Let Assumptions~\ref{ass:proportion} and \ref{ass:step} hold. Define $\tilde{u}_0 := V_{\bot}^T F^{-\frac{1}{2}} u_0$ and $\tilde{\Gamma} := V_{\bot}^T F^{-1} V_{\bot}$. Then the synchronization cost of the system $\hat{T}_{\omega\mathrm{p}}$ is given by
\begin{align*}
\|\tilde{\omega}\|_2^2 = \tilde{u}_0^T \left(\tilde{\Gamma}\circ \tilde{H}\right) \tilde{u}_0,
\end{align*}
where $\circ$ denotes the Hadamard product and $\tilde{H}\in\real^{(n-1) \times (n-1)}$ is the matrix with entries
\begin{equation*}
    \tilde{H}_{kl} := \int_0^\infty h_{\mathrm{u},k}(t) h_{\mathrm{u},l}(t)\ \mathrm{d}t\,,\quad\forall k,l \in \{1,\dots, n-1\}
\end{equation*}
with $\hat{h}_{\mathrm{u},k}(s) := \hat{h}_{\mathrm{p},{k+1},\mathrm{T}}(s)/s$ and $\hat{h}_{\mathrm{p},k,\mathrm{T}}(s)$ being a specified case of the transfer function $\hat{h}_{\mathrm{p},k}(s)$ defined in \eqref{eq:hp-s}, i.e., when the turbine is triggered.
\end{lem}
\begin{proof}
This is a direct extension of \cite[Proposition 2]{pm2019preprint}.
\end{proof}
Lemma~\ref{lem:syncost-generic} shows that the computation of the synchronization cost requires knowing the inner products $\tilde{H}_{kl}$. However, the general expressions of these inner products for an arbitrary combination of $k$ and $l$ are already too tedious to be useful in our analysis. Therefore, we will investigate instead  bounds on the synchronization cost in terms of the inner products $\tilde{H}_{kl}$ when $k=l$; which are exactly the $\mathcal{H}_2$ norms of transfer functions $\hat{h}_{\mathrm{u},k}(s)$.
\begin{lem}
[Bounds for Hadamard product]\label{lem:bounds-Had}
Let $P\in\real^{n\times{}n}$ be a symmetric matrix with minimum and maximum eigenvalues given by $\lambda_{\mathrm{min}}(P)$ and $\lambda_{\mathrm{max}}(P)$, respectively. Then $\forall x, y\in\real^n$,
\[
\lambda_{\mathrm{min}}(P)\sum_{k=1}^nx_k^2y_k^2\leq{}x^T\left(P\circ\left(yy^T\right)\right)x\leq{}\lambda_{\mathrm{max}}(P) \sum_{k=1}^nx_k^2y_k^2.
\]
\end{lem}
\begin{proof}
First note that
\[
\begin{aligned}
x^T\left(P\circ\left(yy^T\right)\right)x&=\tr{P^T\left(x\circ y\right)\left(x\circ y\right)^T}\\
&=\left(x\circ y\right)^T P^T\left(x\circ y\right).
\end{aligned}
\]
Let $w:=x\circ y$. Since $P$ is symmetric, by Rayleigh \cite{Horn2012MA}
\[
\lambda_{\mathrm{min}}(P) w^Tw\leq{}x^T\left(P\circ\left(yy^T\right)\right)x\leq{}\lambda_{\mathrm{max}}(P) w^Tw.
\]
Observing that $w^Tw=\sum_{k=1}^nx_k^2y_k^2$ completes the proof.
\end{proof}

Lemma~\ref{lem:bounds-Had} implies the following bounds on the synchronization cost.

\begin{thm}[Bounds on synchronization cost]\label{thm:bound-cost}
Let Assumptions~\ref{ass:proportion} and \ref{ass:step} hold. Then the synchronization cost of the system $\hat{T}_{ \omega\mathrm{p}}$ is bounded by $\underline{\|\tilde{\omega}\|_2^2} \leq\|\tilde{\omega}\|_2^2 \leq \overline{\|\tilde{\omega}\|_2^2}$, where
\[
\underline{\|\tilde{\omega}\|_2^2}\!\!:=\!\!\frac{\sum_{k=1}^{n-1}\!\tilde{u}_{0,k}^2\| \hat{h}_{\mathrm{u},k} \|_{\mathcal{H}_2}^2}{\max_{i \in \mathcal{V}} \left(f_i \right)}\  \text{and}\  \overline{\|\tilde{\omega}\|_2^2} \!\!:=\!\!\frac{\sum_{k=1}^{n-1}\!\tilde{u}_{0,k}^2\!\| \hat{h}_{\mathrm{u},k} \|_{\mathcal{H}_2}^2}{\min_{i \in \mathcal{V}} \left(f_i \right)} .
\]
\end{thm}

\begin{proof}
By Lemma~\ref{lem:syncost-generic},
\[
\begin{aligned}
\|\tilde{\omega}\|_2^2&\!=\!\!\int_0^\infty{}\tilde{u}_0^T\left(\tilde{\Gamma}\circ{}\left(h_{\mathrm{u}}(t)h_{\mathrm{u}}(t)^T\right)\right)\tilde{u}_0\,\mathrm{d}t\\
\!&\!\!\!\geq{}\!\!\int_0^\infty{}\lambda_{\min}(\tilde{\Gamma})\sum_{k=1}^{n-1}\tilde{u}_{0,k}^2h_{\mathrm{u},k}(t)^2\,\mathrm{d}t\\
\!&\!\!\!=\!\lambda_{\min}(\tilde{\Gamma})\sum_{k=1}^{n-1}\tilde{u}_{0,k}^2\| \hat{h}_{\mathrm{u},k} \|_{\mathcal{H}_2}^2\\
\!&\!\!\!\geq{}\!\lambda_{\min}(F^{-1})\sum_{k=1}^{n-1}\tilde{u}_{0,k}^2\| \hat{h}_{\mathrm{u},k} \|_{\mathcal{H}_2}^2 =\!\frac{\sum_{k=1}^{n-1}\tilde{u}_{0,k}^2\| \hat{h}_{\mathrm{u},k} \|_{\mathcal{H}_2}^2}{\max_{i \in \mathcal{V}} \left(f_i \right)},
\end{aligned}
\]
which concludes the proof of the lower bound.
The first inequality follows from Lemma~\ref{lem:bounds-Had} by setting $P = \tilde{\Gamma}$, $x = \tilde{u}_0$, and $y = h_{\mathrm{u}}(t):=\left(h_{\mathrm{u},k}(t), k \in \{1,\dots, n-1\}\right) \in \real^{n-1}$. The second inequality follows from the interlacing theorem \cite[Theorem 4.3.17]{Horn2012MA}. The upper bound can be proved similarly.
\end{proof}
\begin{rem}[Synchronization cost in homogeneous case]\label{rem:syncost-homo}
In the system with homogeneous parameters, i.e., $F=fI_n$ for some $f>0$, the identical lower and upper bounds on the synchronization cost imply that
\[\|\tilde{\omega}\|_2^2=f^{-1}\sum_{k=1}^{n-1}\!\tilde{u}_{0,k}^2\| \hat{h}_{\mathrm{u},k} \|_{\mathcal{H}_2}^2.
\]
\end{rem}

\subsubsection{Nadir}
A deep Nadir poses a threat to the reliable operation of a power system. Hence one of the goals of inverter control laws is the reduction of Nadir. We seek to evaluate the ability of different control laws to eliminate Nadir. To this end, we provide a necessary and sufficient condition for Nadir elimination in a second-order system with a zero.

\begin{thm}[Nadir elimination for a second-order system] \label{th:nonadir-con}
Assume $K>0$, $z>0$, $\xi \geq 0$, $\omega_\mathrm{n}> 0$. The step response of a second-order system with transfer function given by
\begin{equation*}
    \hat{h}(s) = \dfrac{K\left(s + z\right)}{ s^2 + 2\xi\omega_\mathrm{n} s + \omega_\mathrm{n}^2 }
\end{equation*}
has no Nadir if and only if
\begin{align}\label{eq:nonadir-con}
    1 \leq \xi\leq z/\omega_\mathrm{n} \quad\text{or}\quad
    \begin{cases}
    \xi>z/\omega_\mathrm{n}\\
    \xi  \geq \left(z/\omega_\mathrm{n}+\omega_\mathrm{n}/z\right)/2
    \end{cases},
\end{align}
where the conditions in braces jointly imply $\xi>1$.
\end{thm}
\begin{proof}
 Basically, Nadir must occur at some non-negative finite time instant $t_\mathrm{nadir}$, such that $\dot{p}_\mathrm{u}(t_\mathrm{nadir}) =0$ and $p_\mathrm{u}(t_\mathrm{nadir})$ is a maximum,
where $p_\mathrm{u}(t)$ denotes the unit-step response of $\hat{h}(s)$, i.e., $\hat{p}_\mathrm{u}(s) := \hat{h}(s)/s$. We consider three cases based on the value of damping ratio $\xi$ separately:
\begin{enumerate}
    \item Under damped case ($0\leq\xi<1$): The output is
    \begin{align*}
    \hat{p}_\mathrm{u}(s) = \dfrac{Kz}{\omega_\mathrm{n}^2} \left[\dfrac{1}{s}- \dfrac{s + \xi\omega_\mathrm{n}}{ (s+\xi\omega_\mathrm{n})^2  + \omega_\mathrm{d}^2 }- \dfrac{\xi\omega_\mathrm{n}  - \omega_\mathrm{n}^2 z^{-1}}{ (s+\xi\omega_\mathrm{n})^2  + \omega_\mathrm{d}^2 }\right]
\end{align*}
with $\omega_\mathrm{d} := \omega_\mathrm{n} \sqrt{1 - \xi^2}$, which gives the time domain response
\begin{align*}
    p_\mathrm{u}(t)
    = \dfrac{Kz}{\omega_\mathrm{n}^2} \left[1- e^{-\xi\omega_\mathrm{n} t} \eta_0\sin{(\omega_\mathrm{d} t+\phi)}\right]\;,
\end{align*}
where
\[
\eta_0 = \!\sqrt{1+\dfrac{\left(\xi\omega_\mathrm{n}  - \omega_\mathrm{n}^2 z^{-1}\right)^2}{\omega_\mathrm{d}^2}}\ \text{and}\
\tan\phi = \dfrac{\omega_\mathrm{d}}{\xi\omega_\mathrm{n}  - \omega_\mathrm{n}^2 z^{-1}}.
\]
Clearly, the above response must have oscillations. Therefore, for the case $0\leq\xi<1$, Nadir always exists.
\item Critically damped case ($\xi=1$): The output is
\begin{align*}
    \hat{p}_\mathrm{u}(s)
    =\dfrac{Kz}{\omega_\mathrm{n}^2} \left[\dfrac{1}{s}- \dfrac{1}{ s + \omega_\mathrm{n} }- \dfrac{ \omega_\mathrm{n} - \omega_\mathrm{n}^2 z^{-1}}{ \left(s + \omega_\mathrm{n}\right)^2 }\right]\;,
\end{align*}
which gives the time domain response
\begin{align*}
    p_\mathrm{u}(t) = \dfrac{Kz}{\omega_\mathrm{n}^2} \left\{1- e^{-\omega_\mathrm{n} t}\left[1 + \left(\omega_\mathrm{n}  - \omega_\mathrm{n}^2 z^{-1}\right) t\right]\right\}\;.
\end{align*}
Thus,
\begin{align*}
    \dot{p}_\mathrm{u}(t) = Kz e^{-\omega_\mathrm{n} t}\left[ \left(  1- \omega_\mathrm{n} z^{-1}\right) t +  z^{-1}\right]\;.
\end{align*}
 Letting $\dot{p}_\mathrm{u}(t) =0$ yields
\begin{align*}
    \omega_\mathrm{n} e^{-\omega_\mathrm{n} t} \left[1 + \left(\omega_\mathrm{n}  - \omega_\mathrm{n}^2 z^{-1}\right) t\right] = e^{-\omega_\mathrm{n} t} \left(\omega_\mathrm{n}  - \omega_\mathrm{n}^2 z^{-1}\right)\;,
\end{align*}
which has a non-negative finite solution
\begin{align*}
    t_\mathrm{nadir}
    =\dfrac{z^{-1}}{\omega_\mathrm{n} z^{-1} -1}
\end{align*}
whenever $\omega_\mathrm{n} z^{-1} > 1$. For any $\epsilon>0$, it holds that
\begin{align*}
\dot {p}_\mathrm{u}(t_\mathrm{nadir}-\epsilon) = \epsilon Kz e^{-\omega_\mathrm{n} \left(t_\mathrm{nadir}-\epsilon\right)} \left(\omega_\mathrm{n} z^{-1}-1\right)>0\;,\\
    \dot {p}_\mathrm{u}(t_\mathrm{nadir}+\epsilon) = \epsilon Kz e^{-\omega_\mathrm{n} \left(t_\mathrm{nadir}+\epsilon\right)} \left(  1- \omega_\mathrm{n} z^{-1}\right)<0  \;.
\end{align*}
Clearly, Nadir occurs at $t_\mathrm{nadir}$.
Therefore, for the case $\xi = 1$, Nadir is eliminated if and only if $\omega_\mathrm{n} z^{-1} \leq 1$. To put it more succinctly, we combine the two conditions into
\begin{equation}\label{eq:cri-nonadir}
    1=\xi\leq z/\omega_\mathrm{n}\;.
\end{equation}
\item Over damped case ($\xi>1$): The output is
\begin{align*}
    \hat{p}_\mathrm{u}(s)
    =&\dfrac{Kz}{\omega_\mathrm{n}^2} \left(\dfrac{1}{s}- \dfrac{\eta_1}{ s +\sigma_1}- \dfrac{\eta_2}{ s+\sigma_2}\right)
\end{align*}
with
\begin{align*}
    \sigma_{1,2} = \omega_\mathrm{n}\left(\xi\pm \sqrt{\xi^2-1}\right)\ \ \text{and}\ \ \eta_{1,2} =\dfrac{1}{2}\mp \dfrac{\xi- \omega_\mathrm{n}z^{-1}}{ 2\sqrt{\xi^2-1}}\;,
\end{align*}
which gives the time domain response
\begin{align*}
    p_\mathrm{u}(t) = \dfrac{Kz}{\omega_\mathrm{n}^2} \left(1- \eta_1 e^{-\sigma_1 t}-\eta_2 e^{-\sigma_2 t}\right)\;.
\end{align*}
Thus,
\begin{align*}
    \dot{p}_\mathrm{u}(t) = \dfrac{Kz}{\omega_\mathrm{n}^2} \left(\sigma_1 \eta_1 e^{-\sigma_1 t}+\sigma_2\eta_2 e^{-\sigma_2 t}\right)\;.
\end{align*}
Letting $\dot{p}_\mathrm{u}(t) =0$ yields $\sigma_1 \eta_1 e^{-\sigma_1 t} = - \sigma_2 \eta_2 e^{-\sigma_2 t}$,
which has a non-negative finite solution
\begin{align*}
    t_\mathrm{nadir}
    =\dfrac{1}{2\omega_\mathrm{n}  \sqrt{\xi^2-1}}\ln{\dfrac{1  - \omega_\mathrm{n}z^{-1}\left(\xi+ \sqrt{\xi^2-1}\right)}{1  - \omega_\mathrm{n}z^{-1}\left(\xi- \sqrt{\xi^2-1}\right)}}
\end{align*}
whenever $1  - \omega_\mathrm{n}z^{-1}\left(\xi- \sqrt{\xi^2-1}\right)<0$. For any $\epsilon>0$, it holds that
\begin{align*}
    \dot{p}_\mathrm{u}(t_\mathrm{nadir}-\epsilon)
    >& \dfrac{Kz}{\omega_\mathrm{n}^2} e^{\sigma_1 \epsilon} \left(\sigma_1 \eta_1 e^{-\sigma_1 t_\mathrm{nadir}}+ \sigma_2\eta_2e^{-\sigma_2 t_\mathrm{nadir}}\right)\\
    =& e^{\sigma_1 \epsilon} \dot{p}_\mathrm{u}(t_\mathrm{nadir})=0\;,\\
    \dot{p}_\mathrm{u}(t_\mathrm{nadir}+\epsilon)
    <& \dfrac{Kz}{\omega_\mathrm{n}^2} e^{-\sigma_1 \epsilon} \left(\sigma_1 \eta_1 e^{-\sigma_1 t_\mathrm{nadir}}+ \sigma_2\eta_2e^{-\sigma_2 t_\mathrm{nadir}}\right)\\
    =& e^{-\sigma_1 \epsilon} \dot{p}_\mathrm{u}(t_\mathrm{nadir})=0\;,
\end{align*}
since $\sigma_1>\sigma_2>0$ and one can show that $\sigma_2\eta_2<0$. Clearly, Nadir occurs at $t_\mathrm{nadir}$.
Therefore, for the case $\xi > 1$, Nadir is eliminated if and only if $1  - \omega_\mathrm{n}z^{-1}\left(\xi- \sqrt{\xi^2-1}\right)\geq0$, i.e., $\sqrt{\xi^2-1} \geq \xi-z/\omega_\mathrm{n}$,
which holds if and only if
\begin{align*}
    \xi\leq z/\omega_\mathrm{n} \quad\text{or}\quad
    \begin{cases}
    \xi> z/\omega_\mathrm{n}\\
    \xi  \geq \left(z/\omega_\mathrm{n}+\omega_\mathrm{n}/z\right)/2
    \end{cases}.
\end{align*}
Thus we get the conditions
\begin{align}\label{eq:ov-nonadir}
    1 < \xi\leq z/\omega_\mathrm{n} \quad\text{or}\quad
    \begin{cases}
    \xi>1\\
    \xi>z/\omega_\mathrm{n}\\
    \xi  \geq \left(z/\omega_\mathrm{n}+\omega_\mathrm{n}/z\right)/2
    \end{cases}.
\end{align}
\end{enumerate}
Finally, since $\forall a, b \geq 0$, $(a+b)/2\geq\sqrt{ab}$ with equality only when $a=b$, it follows that the second condition in \eqref{eq:ov-nonadir} can only hold when $\xi>1$.
Thus we can combine
 \eqref{eq:cri-nonadir} and \eqref{eq:ov-nonadir} to yield \eqref{eq:nonadir-con}.
\end{proof}

%% file: 04-the_need_of_a_better_solution.tex
We now apply the results in Section~\ref{sec:result} to illustrate the performance limitations of the traditional control laws DC and VI.
%
%
%
With this aim, we seek
 to quantify the frequency variance \eqref{eq:h2_def_E} under DC and VI through the $\mathcal{H}_2$ norm of $\hat{T}_{\omega\mathrm{dn},\mathrm{DC}}$ and $\hat{T}_{\omega\mathrm{dn},\mathrm{VI}}$, as well as the steady-state effort share \eqref{eq:ES}, synchronization cost~\eqref{eq:sync_cost}, and Nadir \eqref{eq:Nadir} through the step response characterizations of $\hat{T}_{\omega\mathrm{p,DC}}$ and $\hat{T}_{\omega \mathrm{p,VI}}$.

\subsection{Steady-state Effort Share}
\begin{cor}[Synchronous frequency under DC and VI]\label{lem:syn-fre-dc}
Let Assumption~\ref{ass:step} hold. When $q_{\mathrm{r},i}$ is defined by the control law  DC \eqref{eq:dy-dc} or VI \eqref{eq:dy-vi}, the steady-state frequency deviation of the system $\hat{T}_{\omega \mathrm{p,DC}}$ or $\hat{T}_{\omega \mathrm{p, VI}}$ synchronizes to the synchronous frequency, i.e., $\omega_{\mathrm{ss}} = \omega_{\mathrm{syn}} \mathbbold{1}_n$ with
\begin{equation}
 \omega_{\mathrm{syn}} = \dfrac{\sum_{i=1}^n u_{0,i}}{\sum_{i=1}^n \left( d_i + r_{\mathrm{t},i}^{-1} + r_{\mathrm{r},i}^{-1} \right)}\;. \label{eq:ome-syn-dc}
\end{equation}
\end{cor}
\begin{proof}
The result follows directly from Lemma~\ref{lem:syn-fre}.
\end{proof}

Now, the corollary below gives the expression for the steady-state effort share when inverters are under the control law DC or VI.

\begin{cor}[Steady-state effort share of DC and VI]\label{thm:ss-DC}
Let Assumption~\ref{ass:step} hold. If $q_{\mathrm{r},i}$ is under the control law \eqref{eq:dy-dc} or \eqref{eq:dy-vi}, then the steady-state effort share of the system $\hat{T}_{\omega \mathrm{p, DC}}$ or $\hat{T}_{\omega \mathrm{p, VI}}$ is given by
\begin{equation}\label{eq:es-ratio}
    \mathrm{ES} = \frac{\sum_{i=1}^n r_{\mathrm{r},i}^{-1}}{\sum_{i=1}^n  \left( d_i + {r_{\mathrm{t},i}^{-1} + r_{\mathrm{r},i}^{-1}} \right) }\;.
\end{equation}
\end{cor}
\begin{proof}
The result follows directly from Theorem~\ref{thm:ss-es} applied to \eqref{eq:dy-dc} and \eqref{eq:dy-vi}.
\end{proof}

Corollary~\ref{thm:ss-DC} indicates that DC and VI have the same steady-state effort share, which increases as $r_{\mathrm{r},i}^{-1}$ increase.
 However, $r_{\mathrm{r},i}^{-1}$ are parameters that also directly affect the dynamic performance of the power system, which can be seen clearly from the dynamic performance analysis.

\subsection{Power Fluctuations and Measurement Noise}\label{ssec:VI-dy}


Using Theorem~\ref{thm:h2-sum} and Lemma~\ref{lm:h2-4th}, it is possible to get closed form expressions of $\mathcal{H}_2$ norms for systems $\hat{T}_{\omega\mathrm{dn},\mathrm{DC}}$ and $\hat{T}_{\omega\mathrm{dn},\mathrm{VI}}$.

\begin{cor}[Frequency variance under DC and VI]\label{thm:noise-VI}
Let Assumptions~\ref{ass:proportion} and \ref{ass:noise} hold. The squared $\mathcal{H}_2$ norm of $\hat{T}_{\omega\mathrm{dn},\mathrm{DC}}$ and $\hat{T}_{\omega\mathrm{dn},\mathrm{VI}}$ is given by
\begin{subequations}
\begin{equation}\label{eq:noise-DC}
	\|\hat{T}_{\omega\mathrm{dn}, \mathrm{DC}}\|_{\mathcal{H}_2}^2 = \sum_{k=1}^n \Gamma_{kk} \dfrac{\kappa_\mathrm{p}^2 + r_\mathrm{r}^{-2} \kappa_\omega^2}{2m \check{d}}  ,
\end{equation}
\begin{equation}\label{eq:noise-VI}
\|\hat{T}_{\omega\mathrm{dn}, \mathrm{VI}}\|_{\mathcal{H}_2}^2 = \infty \;,
\end{equation}
\end{subequations}
respectively, where $\check{d} := d + r_\mathrm{r}^{-1}$.
\end{cor}

\begin{proof}
We study the two cases separately.

We begin with $\|\hat{T}_{\omega\mathrm{dn}, \mathrm{DC}}\|_{\mathcal{H}_2}^2$.
Applying \eqref{eq:go-sw} and \eqref{eq:co-dc} to \eqref{eq:hp-s} and \eqref{eq:homega-s} shows
$\hat{h}_{{\mathrm{p},k},\mathrm{DC}}(s)$ is a transfer function with $b_4=a_0=b_0=a_1=b_1=0, a_2 = \lambda_k/m, b_2 = 0, a_3 = \check{d}/m, b_3 =1/m$, while $\hat{h}_{{\omega,k},\mathrm{DC}}(s)$ is a transfer function with $b_4=a_0=b_0=a_1=b_1=0, a_2 = \lambda_k/m, b_2 = 0, a_3 = \check{d}/m, b_3 =-r_\mathrm{r}^{-1}/m$.
Thus, by Lemma~\ref{lm:h2-4th},
\begin{align*}
\|\hat{h}_{{\mathrm{p},k},\mathrm{DC}}\|_{\mathcal{H}_2}^2=\frac{1}{2m\check{d}} \quad \text{and} \quad
\|\hat{h}_{{\omega,k},\mathrm{DC}}\|_{\mathcal{H}_2}^2=\frac{r_\mathrm{r}^{-2}}{2m\check{d}}\;.
\end{align*}
Then \eqref{eq:noise-DC} follows from Theorem~\ref{thm:h2-sum}.

We now turn to show that $\|\hat{T}_{\omega\mathrm{dn}, \mathrm{VI}}\|_{\mathcal{H}_2}^2$ is infinite. Applying \eqref{eq:go-sw} and \eqref{eq:co-vi} to \eqref{eq:homega-s} yields
\begin{align*}
\hat{h}_{{\omega,k},\mathrm{VI}}(s) =& - \frac{m_{\mathrm{v}} s^2 + r_\mathrm{r}^{-1} s}{(m + m_\mathrm{v}) s^2 + \check{d} s + \lambda_k}\;,
\end{align*}
which by Lemma~\ref{lm:h2-4th} has $b_4 = -m_\mathrm{v}/\left(m + m_\mathrm{v}\right)\neq0$ and thus $\|\hat{h}_{{\omega,k},\mathrm{DC}}\|_{\mathcal{H}_2}^2=\infty$. Then \eqref{eq:noise-VI} follows directly from Theorem~\ref{thm:h2-sum}. 
\end{proof}
\begin{cor}[Optimal $r_\mathrm{r}^{-1}$ for $\|\hat{T}_{\omega\mathrm{dn}, \mathrm{DC}}\|_{\mathcal{H}_2}^2$] \label{cor:optimal-h2-dc}Let Assumptions~\ref{ass:proportion} and \ref{ass:noise} hold. Then
\begin{align}\label{eq:rr-star}
    r_\mathrm{r}^{-1\star}\!\!:=\!\argmin_{r_\mathrm{r}^{-1} > 0}\! \|\hat{T}_{\omega\mathrm{dn}, \mathrm{DC}}\|_{\mathcal{H}_2}^2\!\!=\!-d +\! \sqrt{d^2 + (\kappa_{\mathrm{p}}/\kappa_{\omega})^2}\,.
\end{align}
\end{cor}
\begin{proof}
The partial derivative of $\|\hat{T}_{\omega\mathrm{dn}, \mathrm{DC}}\|_{\mathcal{H}_2}^2$ with respect to $r_\mathrm{r}^{-1}$ is
\begin{align}\label{eq:h2-dc-partial}
\partial_{r_\mathrm{r}^{-1}}\|\hat{T}_{\omega\mathrm{dn}, \mathrm{DC}}\|_{\mathcal{H}_2}^2 = \sum_{k=1}^n \Gamma_{kk} \frac{\kappa_\omega^2 r_\mathrm{r}^{-2} \!+\! 2d\kappa_\omega^2 r_\mathrm{r}^{-1}\!-\!\kappa_\mathrm{p}^2}{2m\check{d}^2}\,.
\end{align}
By equating \eqref{eq:h2-dc-partial} to 0, we can solve the corresponding $r_\mathrm{r}^{-1}$ as ${r_\mathrm{r}^{-1\star}}_\pm=-d \pm \sqrt{d^2 + (\kappa_{\mathrm{p}}/\kappa_{\omega})^2}$. The only positive root is therefore $r_\mathrm{r}^{-1\star}:=-d + \sqrt{d^2 + (\kappa_{\mathrm{p}}/\kappa_{\omega})^2}$. We now show that $\Gamma_{kk} > 0$, $\forall  k \in \{1,\dots, n\}$. Recall that $\Gamma := V^T F^{-1} V$. We know $\Gamma_{kk} = \sum_{j=1}^n ( v_{k,j}^2/f_j)$. Since $v_k$ is an eigenvector, $\forall k \in \{1,\dots, n\}$, there must exist at least one $j\in \mathcal{V}$ such that $v_{k,j}\neq0$. Since $f_i > 0$, $\forall i$, we have that $\Gamma_{kk} > 0$, $\forall  k \in \{1,\dots, n\}$. In addition, since the denominator of \eqref{eq:h2-dc-partial} is always positive and the highest order coefficient of the numerator is positive, whenever $0 < r_\mathrm{r}^{-1} < r_\mathrm{r}^{-1\star}$, then $ \partial_{r_\mathrm{r}^{-1}}\|\hat{T}_{\omega\mathrm{dn}, \mathrm{DC}}\|_{\mathcal{H}_2}^2 < 0$, and if $r_\mathrm{r}^{-1} > r_\mathrm{r}^{-1\star}$, then $ \partial_{r_\mathrm{r}^{-1}}\|\hat{T}_{\omega\mathrm{dn}, \mathrm{DC}}\|_{\mathcal{H}_2}^2 > 0$. Therefore, $r_\mathrm{r}^{-1\star}$ is the minimizer of $\|\hat{T}_{\omega\mathrm{dn}, \mathrm{DC}}\|_{\mathcal{H}_2}^2$.
\end{proof}

Two main observations can be made from Corollary~\ref{thm:noise-VI}. First, the control parameter $r_\mathrm{r}^{-1}$ of DC has an direct effect on the size of the frequency variance in the system, which makes it impossible to require DC to bear an assigned amount of steady-state effort share and reduce the frequency variance at the same time. The other important point is that VI will induce unbounded frequency variance, which poses a threat to the operation of the power system. Therefore, neither DC nor VI is good solution to improve the frequency variance without sacrificing the steady-state effort share.
\subsection{Synchronization Cost}
Theorem~\ref{thm:bound-cost} implies that the synchronization cost of $\hat{T}_{\omega \mathrm{p, DC}}$ and $\hat{T}_{\omega \mathrm{p, VI}}$ are bounded by a weighted sum of $\| \hat{h}_{\mathrm{u},k,\mathrm{DC}} \|_{\mathcal{H}_2}^2$ and $\| \hat{h}_{\mathrm{u},k,\mathrm{VI}} \|_{\mathcal{H}_2}^2$, respectively. Hence, in order to see the limited ability of DC and VI to reduce the synchronization cost, we need to gain a deeper understanding of $\| \hat{h}_{\mathrm{u},k,\mathrm{DC}} \|_{\mathcal{H}_2}^2$ and $\| \hat{h}_{\mathrm{u},k,\mathrm{VI}} \|_{\mathcal{H}_2}^2$ first.

\begin{thm}[Bounds of $\| \hat{h}_{\mathrm{u},k,\mathrm{DC}} \|_{\mathcal{H}_2}^2$ and $\| \hat{h}_{\mathrm{u},k,\mathrm{VI}} \|_{\mathcal{H}_2}^2$]
\label{lem:bounds-VI}
Let Assumptions~\ref{ass:proportion} and \ref{ass:step} hold. Then, given $r_\mathrm{r}^{-1}>0$, $\forall m_\mathrm{v} > 0$,
\begin{align*}
 \dfrac{1}{2 \lambda_{k+1} \!\left(\check{d} \!+\! r_\mathrm{t}^{-1}\right)} \!\!<\! \| \hat{h}_{\mathrm{u},k,\mathrm{VI}} \|_{\mathcal{H}_2}^2 \!\!<\! \| \hat{h}_{\mathrm{u},k,\mathrm{DC}} \|_{\mathcal{H}_2}^2 \!\!<\! \| \hat{h}_{\mathrm{u},k,\mathrm{SW}} \|_{\mathcal{H}_2}^2,
\end{align*}
where $\| \hat{h}_{\mathrm{u},k,\mathrm{SW}} \|_{\mathcal{H}_2}^2$ represents the inner products of the open-loop system with no additional control from inverters. 
\end{thm}
\begin{proof}
Considering that DC can be viewed as VI with $m_\mathrm{v} = 0$ and the open-loop system can be viewed as VI with $m_\mathrm{v}=r_\mathrm{r}^{-1}=0$, we only compute $\|\hat{h}_{\mathrm{u},k,\mathrm{VI}}\|_{\mathcal{H}_2}^2$, which straightforwardly implies the other two. Applying \eqref{eq:go-sw-tb} and \eqref{eq:co-vi} to \eqref{eq:hp-s} shows $\hat{h}_{\mathrm{u},k,\mathrm{VI}}(s)=\hat{h}_{\mathrm{p},k+1,\mathrm{T,VI}}(s)/s$ is a transfer function with $b_4=a_0=b_0=0, a_1=\lambda_{k+1}/\left(\check{m} \tau\right), b_1=1/\left(\check{m} \tau\right), a_2 = \left(\check{d} + r_{\mathrm{t}}^{-1} + \lambda_{k+1} \tau\right)/\left(\check{m} \tau\right), b_2 = 1/\check{m}, a_3 = \left(\check{m} + \check{d} \tau\right)/\left(\check{m} \tau\right), b_3 =0$.
Then it follows from Lemma~\ref{lm:h2-4th} that
\begin{align} 
\|\hat{h}_{\mathrm{u},k,\mathrm{VI}}\|_{\mathcal{H}_2}^2 
    %
   \!\!=\!\dfrac{\check{m} + \tau\! \left(\lambda_{k+1} \tau + \check{d} \right)}{ 2 \lambda_{k+1}\!\left[\tau \check{d} \left(\lambda_{k+1} \tau + \check{d} +\! r_\mathrm{t}^{-1}\right) \!+ \!\check{m}\!\left(\check{d} + r_\mathrm{t}^{-1}\right)\right] }. \nonumber
\end{align}
Since $\| \hat{h}_{\mathrm{u},k,\mathrm{VI}} \|_{\mathcal{H}_2}^2$ is a function of $r_\mathrm{r}^{-1}$ and $m_\mathrm{v}$, in what follows we denote it by $\rho(r_\mathrm{r}^{-1}, m_\mathrm{v})$. In order to have an insight on how $\| \hat{h}_{\mathrm{u},k,\mathrm{VI}} \|_{\mathcal{H}_2}^2$ changes with $r_\mathrm{r}^{-1}$ and $m_\mathrm{v}$, we take partial derivatives of $\rho(r_\mathrm{r}^{-1}, m_\mathrm{v})$ with respect to $r_\mathrm{r}^{-1}$ and $m_\mathrm{v}$, i.e., 
\begin{align}
	&\partial_{r_\mathrm{r}^{-1}} \rho (r_\mathrm{r}^{-1}, m_\mathrm{v}) \nonumber\\
    =& - \!\dfrac{
    \left[\check{m}+ \tau \left(\lambda_{k+1} \tau + \check{d} \right)\right]^2 + \lambda_{k+1} \tau^3  r_\mathrm{t}^{-1} }{ 2 \lambda_{k+1}\left[ \tau \check{d} \left(\lambda_{k+1} \tau + \check{d} + r_\mathrm{t}^{-1}\right)  + \check{m}(\check{d} + r_\mathrm{t}^{-1}) \right]^2}\;,\nonumber\\
    &\partial_{m_\mathrm{v}} \rho (r_\mathrm{r}^{-1}, m_\mathrm{v}) \nonumber\\
    =& - \!\dfrac{ \tau^2 r_\mathrm{t}^{-1}}{ 2 \left[ \tau \check{d} \left(\lambda_{k+1} \tau + \check{d} + r_\mathrm{t}^{-1}\right) + \check{m}(\check{d} + r_\mathrm{t}^{-1})  \right]^2}\;. \nonumber
\end{align}
Clearly, for all $r_\mathrm{r}^{-1} \geq 0$, $\partial_{r_\mathrm{r}^{-1}} \rho (r_\mathrm{r}^{-1}, m_\mathrm{v}) < 0$, which means that $\rho (r_\mathrm{r}^{-1}, m_\mathrm{v})$ is a monotonically decreasing function of $r_\mathrm{r}^{-1}$. Similarly, for all $m_\mathrm{v} \geq 0$, $\partial_{m_\mathrm{v}} \rho (r_\mathrm{r}^{-1}, m_\mathrm{v}) < 0$, which means that $\rho (r_\mathrm{r}^{-1}, m_\mathrm{v})$ is a monotonically decreasing function of $m_\mathrm{v}$. Therefore, given $r_\mathrm{r}^{-1} > 0$, $\forall m_\mathrm{v} > 0$, it holds that 
\[
\lim_{m_\mathrm{v}\to \infty} \rho(r_\mathrm{r}^{-1}, m_\mathrm{v})<\rho(r_\mathrm{r}^{-1}, m_\mathrm{v}) < \rho(r_\mathrm{r}^{-1}, 0)< \rho(0, 0)\,.
\]
Recall that $\| \hat{h}_{\mathrm{u},k,\mathrm{VI}} \|_{\mathcal{H}_2}^2 = \rho(r_\mathrm{r}^{-1}, m_\mathrm{v})$, $\| \hat{h}_{\mathrm{u},k,\mathrm{DC}} \|_{\mathcal{H}_2}^2 = \rho(r_\mathrm{r}^{-1}, 0)$, and $\| \hat{h}_{\mathrm{u},k,\mathrm{SW}} \|_{\mathcal{H}_2}^2 = \rho(0, 0)$. 
The result follows.
\end{proof}
\begin{cor}[Comparison of synchronization cost in homogeneous case] \label{cor:syncost-homoe-bounds}
Denote the synchronization cost of the open-loop system as $\|\tilde{\omega}_\mathrm{SW}\|_2^2$. Then, under Assumptions~\ref{ass:proportion} and \ref{ass:step}, given $r_\mathrm{r}^{-1}>0$, $\forall m_\mathrm{v} > 0$, we can order the synchronization cost when $F=fI_n$ as: 
\[
\frac{\sum_{k=1}^{n-1}\left(\tilde{u}_{0,k}^2/\lambda_{k+1}\right)}{2f\left(\check{d} + r_\mathrm{t}^{-1}\right)}< \|\tilde{\omega}_\mathrm{VI}\|_2^2 < \|\tilde{\omega}_\mathrm{DC}\|_2^2 < \|\tilde{\omega}_\mathrm{SW}\|_2^2\,.
\]
\end{cor}
\begin{proof}
The result follows by combining Remark~\ref{rem:syncost-homo} and Theorem~\ref{lem:bounds-VI}.
\end{proof}
\begin{cor}[Lower bound of synchronization cost under DC and VI]\label{cor:low-bound-pro}
Under Assumptions~\ref{ass:proportion} and \ref{ass:step}, the ordering of the size of the bounds on the synchronization cost of open-loop, DC, and VI depends on the parameter values. Thus we cannot order $\|\tilde{\omega}_\mathrm{VI}\|_2^2$, $\|\tilde{\omega}_\mathrm{DC}\|_2^2$, and $\|\tilde{\omega}_\mathrm{SW}\|_2^2$ strictly. Instead, we highlight that, given $r_\mathrm{r}^{-1}>0$, the synchronization cost under DC and VI are bounded below by 
\[
\frac{\sum_{k=1}^{n-1}\left(\tilde{u}_{0,k}^2/\lambda_{k+1}\right)}{2\max_{i \in \mathcal{V}} \left(f_i \right)\left(\check{d} + r_\mathrm{t}^{-1}\right)}\,.
\]
\end{cor}
\begin{proof}
The result follows from Theorems~\ref{thm:bound-cost} and~\ref{lem:bounds-VI}.
\end{proof}
Corollary~\ref{cor:syncost-homoe-bounds} provides both upper and lower bounds for the synchronization cost under DC and VI in homogeneous case. The upper bound verifies that DC and VI do reduce the synchronization cost by adding damping and inertia while the lower bound indicates that the reduction of the synchronization cost through DC and VI is limited by certain value that is dependent on $r_\mathrm{r}^{-1}$. Corollary~\ref{cor:low-bound-pro} implies that in the proportional case the synchronization cost under DC and VI is also bounded below by a value that is dependent on $r_\mathrm{r}^{-1}$. The fact that the lower bound of the synchronization cost under DC and VI is reduced as $r_\mathrm{r}^{-1}$ increases is not satisfactory, since, from the stead-state effort share point of view, a smaller $r_\mathrm{r}^{-1}$ is preferred. However, given a small $r_\mathrm{r}^{-1}$, even if the inertia is very high, i.e., $m_\mathrm{v}\to\infty$, the synchronization cost $\|\tilde{\omega}_\mathrm{VI}\|_2^2$ can never reach zero, not to mention $\|\tilde{\omega}_\mathrm{DC}\|_2^2$.
\subsection{Nadir}
Finally, with the help of Theorem~\ref{th:nonadir-con}, we can determine the conditions that the parameters of DC and VI must satisfy to eliminate Nadir of the system frequency.
\begin{thm}[Nadir elimination under DC and VI]\label{thm:no-nadir-cond-VI}
Under Assumptions~\ref{ass:proportion} and \ref{ass:step}: 
\begin{itemize}
\item
for $\hat{T}_{\omega \mathrm{p, DC}}$, 
the tuning region that eliminates Nadir through DC is $r_{\mathrm{r}}^{-1}$ such that
\begin{align}\label{eq:nadir-DC}
    r_{\mathrm{r}}^{-1} \leq  m\left(\tau^{-1} - 2\sqrt{\tau^{-1}r_{\mathrm{t}}^{-1}/m}\right)-d\;;
\end{align}
\item
for $\hat{T}_{\omega \mathrm{p, VI}}$,
the tuning region that eliminates Nadir through VI is $(r_{\mathrm{r}}^{-1}, m_{\mathrm{v}})$ such that 
\begin{align}\label{eq:nadir-VI}
     r_{\mathrm{r}}^{-1} \!\leq \! \left(m\!+\! m_{\mathrm{v}}\right)\!\left(\tau^{-1}\! -\! 2\sqrt{\tau^{-1}r_{\mathrm{t}}^{-1}\!/\!\left(m\!+\! m_{\mathrm{v}}\right)}\right)\!-d\;.
\end{align}
\end{itemize}
\end{thm}

\begin{proof}
We start by deriving the Nadir elimination condition for VI.
The system frequency of $\hat{T}_{\omega \mathrm{p, VI}}$ is given by \cite{p2017ccc}
\begin{equation*}
	\bar{\omega}_\mathrm{VI}(t) =  \dfrac{\sum_{i=1}^n u_{0,i} }{ \sum_{i=1}^n f_i } p_\mathrm{u,VI}(t)\;,
\end{equation*}
where $p_\mathrm{u,VI}(t)$ is the unit-step response of $\hat{h}_{{\mathrm{p},1},\mathrm{T, VI}}(s)$. Clearly, as long as $p_\mathrm{u,VI}(t)$ has no Nadir, neither does $\bar{\omega}_\mathrm{VI}(t)$. Thus, as shown later, the core is to apply Theorem~\ref{th:nonadir-con} to $\hat{h}_{{\mathrm{p},1},\mathrm{T, VI}}(s)$.
Substituting \eqref{eq:go-sw-tb} and \eqref{eq:co-vi} to \eqref{eq:hp-s} yields
\begin{align*}
\hat{h}_{\mathrm{p},1,\mathrm{T,VI}}(s)
= \frac{1}{\check{m}}\dfrac{s + \tau^{-1}}{ s^2 + 2\xi\omega_\mathrm{n} s + \omega_\mathrm{n}^2 }\;,\nonumber
\end{align*}
where $\omega_\mathrm{n} := \sqrt{\cfrac{\check{d} +r_{\mathrm{t}}^{-1}}{\check{m}\tau}}\;,\quad
    \xi := \dfrac{\tau^{-1}+\check{d}/\check{m}}{2\sqrt{\left(\check{d} +r_{\mathrm{t}}^{-1}\right)/\left(\check{m}\tau\right)}}\;.$
Now we are ready to search the Nadir elimination tuning region by means of Theorem~\ref{th:nonadir-con}. An easy computation shows the following inequality: $2\xi\omega_\mathrm{n} - \tau^{-1} = \check{d}/\check{m} < \left(\check{d} +r_{\mathrm{t}}^{-1}\right)/\check{m}=\omega^2_\mathrm{n} \tau$.
Equivalently, it holds that $\xi < \left[1/\left(\omega_\mathrm{n}\tau\right)+ \omega_\mathrm{n} \tau\right]/2$, 
which indicates that the second set of conditions in \eqref{eq:nonadir-con} cannot be satisfied. Hence, we turn to the first set of conditions in \eqref{eq:nonadir-con}, which holds if and only $\xi \geq 1$ and $\xi\omega_\mathrm{n} \leq \tau^{-1}$. Via simple algebraic computations, this is equivalent to 
\begin{align}\label{eq:cond-nonadir}
\tau \check{d}^2 /\check{m} - 2  \check{d} + \tau^{-1}\check{m} - 4r_{\mathrm{t}}^{-1} \!\geq 0 \quad\text{and}\quad
\check{d}/\check{m} \!\leq \tau^{-1}.
\end{align}
The first condition in \eqref{eq:cond-nonadir} can be viewed as a quadratic inequality with respect to $\check{d}$, which holds if and only if 
\begin{align*}
    \check{d} \leq  \check{m}\left(\tau^{-1} - 2\sqrt{ \cfrac{r_{\mathrm{t}}^{-1}}{\check{m}\tau}}\right)
    \quad\text{or}\quad
    \check{d} \geq  \check{m}\left(\tau^{-1} + 2\sqrt{ \cfrac{r_{\mathrm{t}}^{-1}}{\check{m}\tau}}\right)\,.
\end{align*}
However, only the former region satisfies the second condition in \eqref{eq:cond-nonadir}. This concludes the proof of the second statement. The first statement follows trivially by setting $m_{\mathrm{v}}=0$.
\end{proof}


Important inferences can be made from Theorem~\ref{thm:no-nadir-cond-VI}. The fact that a small $m$ tends to make the term on the right hand side of \eqref{eq:nadir-DC} negative implies that in a low-inertia power system it is impossible to eliminate Nadir using only DC. Undoubtedly, the addition of $m_\mathrm{v}$ makes the tuning region in \eqref{eq:nadir-VI} more accessible, which indicates that VI can help a low-inertia power system improve Nadir.

We end this section by summarizing the pros and cons of each controller.
\begin{itemize}
    \item \textbf{Droop control:}~With only one parameter $r_{\mathrm{r}}^{-1}$, DC can neither reduce frequency variance or synchronization cost without affecting steady-state effort share. Moreover, for low-inertia systems, DC cannot eliminate Nadir.
    \item \textbf{Virtual inertia:}~ VI can use its additional dynamic parameter $m_\mathrm{v}$ to eliminate system Nadir and relatively improve synchronization cost. However this comes at the price of introducing large frequency variance in response to noise, and cannot be decoupled from increases in the steady-state effort share.
\end{itemize}

%% file: 05-idroop.tex
We now show how, by moving away from the broadly proposed approach of mimicking generators response, one can overcome the weaknesses presented in the previous section.
With this aim, we introduce an alternative dynam-i-c Droop (iDroop) controller that uses dynamic feedback to make a trade-off among the several different objectives described in Section \ref{ssec:metrics}. The proposed solution is described below.
\begin{dyn-i}[Dynamic Droop Control]
The dynamics of an inverter with iDroop is given by the transfer function
\begin{equation} \label{eq:dy-idroop}
\hat{c}_i(s) = -\dfrac{\nu_i s + \delta_i r_{\mathrm{r},i}^{-1}}{s + \delta_i} \;,
\end{equation}
where $\delta_i>0$ and $\nu_i>0$ are tunable parameters.
\end{dyn-i}


Similarly to \eqref{eq:go-sw} and \eqref{eq:go-sw-tb}, one can define a representative iDroop inverter controller as 
\begin{equation}\label{eq:co-idroop}
\hat{c}_\mathrm{o}(s) = -\dfrac{\nu s + \delta r_\mathrm{r}^{-1}}{s + \delta}
\end{equation}
with $\nu_i=f_i\nu$, $r_{\mathrm{r},i}=r_{\mathrm{r}}/f_i$, and $\delta_i=\delta$.

In the rest of this section, we expose iDroop to the same performance analysis done for DC and VI in Section \ref{sec:need}. 
\subsection{Steady-state Effort Share}
We can show that iDroop is able to preserve the steady-state behavior given by DC and VI.
\begin{cor}[Synchronous frequency under iDroop]\label{lem:syn-fre-idroop}
Let Assumption~\ref{ass:step} hold. If $q_{\mathrm{r},i}$ is under the control law \eqref{eq:dy-idroop}, then the steady-state frequency deviation of the system $\hat{T}_{\omega \mathrm{p, iDroop}}$  synchronizes to the synchronous frequency given by \eqref{eq:ome-syn-dc}.
\end{cor}

\begin{proof}
The result follows directly from Lemma~\ref{lem:syn-fre}.
\end{proof}
\begin{cor}[Steady-state effort share of iDroop]\label{thm:ss-idroop}
Let Assumption~\ref{ass:step} hold. If $q_{\mathrm{r},i}$ is under the control law \eqref{eq:dy-idroop}, then the steady-state effort share of the system $\hat{T}_{\omega \mathrm{p, iDroop}}$ is given by \eqref{eq:es-ratio}.
\end{cor}
\begin{proof}
The result follows directly from Theorem~\ref{thm:ss-es} applied to \eqref{eq:dy-idroop}.
\end{proof}

Corollaries~\ref{lem:syn-fre-idroop} and~\ref{thm:ss-idroop} suggest that iDroop achieves the same synchronous frequency and steady-state effort share as DC and VI do, which depend on $r_{\mathrm{r},i}^{-1}$. Note that besides $r_{\mathrm{r},i}^{-1}$ iDroop provides us with two more degrees of freedom by $\delta_i$ and $\nu_i$. 

\subsection{Power Fluctuations and Measurement Noise}


The next theorem quantifies the frequency variance under iDroop through the squared $\mathcal{H}_2$ norm of the system $\hat{T}_{\omega\mathrm{dn}, \mathrm{iDroop}}$.

\begin{cor}[Frequency variance under iDroop]\label{thm:noise-idroop}
Let Assumptions~\ref{ass:proportion} and \ref{ass:noise} hold. The squared $\mathcal{H}_2$ norm of $\hat{T}_{\omega\mathrm{dn}, \mathrm{iDroop}}$ is given by
\begin{align}
    & \|\hat{T}_{\omega\mathrm{dn}, \mathrm{iDroop}}\|^2_{\mathcal{H}_2}\label{eq:h2-idroop}\\&= \sum_{k=1}^n \Gamma_{kk} {\dfrac{(\kappa_\mathrm{p}^2 + r_\mathrm{r}^{-2} \kappa_\omega^2) m \delta^2 + (\kappa_\mathrm{p}^2+\nu^2 \kappa_\omega^2)\left(\check{d} \delta + \lambda_k \right) }{2m\left[\check{d}m\delta^2 + (d+\nu)\left(\check{d}\delta + \lambda_k \right)\right]}}.\nonumber
\end{align}
\end{cor}

\begin{proof}
The proof is based on the Theorem~\ref{thm:h2-sum} and Lemma~\ref{lm:h2-4th}. 
Applying \eqref{eq:go-sw} and \eqref{eq:co-idroop} to \eqref{eq:hp-s} and \eqref{eq:homega-s} shows
$\hat{h}_{{\mathrm{p},k},\mathrm{iDroop}}(s)$ is a transfer function with $b_4=a_0=b_0=0, a_1 = \left(\lambda_k \delta\right)/m, b_1 =0, a_2 = \left(\check{d} \delta + \lambda_k\right)/m, b_2 = \delta/m, a_3 = \left(m \delta + d +\nu\right)/m, b_3 =1/m$, while $\hat{h}_{{\omega,k},\mathrm{iDroop}}(s)$ is a transfer function with $b_4=a_0=b_0=0, a_1 = \left(\lambda_k \delta\right)/m, b_1 =0, a_2 = \left(\check{d} \delta + \lambda_k\right)/m, b_2 = -\left(r_\mathrm{r}^{-1}\delta\right)/m, a_3 = \left(m \delta + d +\nu\right)/m, b_3 =-\nu/m$.
Thus, by Lemma~\ref{lm:h2-4th},
\begin{align*}
\|\hat{h}_{{\mathrm{p},k},\mathrm{iDroop}}\|_{\mathcal{H}_2}^2={\dfrac{ m \delta^2 + \check{d} \delta + \lambda_k }{2m\left[\check{d}m\delta^2 + (d+\nu)\left(\check{d} \delta + \lambda_k \right)\right]}} \;,\\
\|\hat{h}_{{\omega,k},\mathrm{iDroop}}\|_{\mathcal{H}_2}^2={\dfrac{ r_\mathrm{r}^{-2} m \delta^2 + \nu^2 \left(\check{d} \delta + \lambda_k \right) }{2m\left[\check{d}m\delta^2 + (d+\nu)\left(\check{d} \delta + \lambda_k \right)\right]}}\;.
\end{align*}
Then \eqref{eq:h2-idroop} follows from Theorem~\ref{thm:h2-sum}.
\end{proof}

The explicit expression of $\|\hat{T}_{\omega\mathrm{dn}, \mathrm{iDroop}}\|^2_{\mathcal{H}_2}$ given in Corollary~\ref{thm:noise-idroop} is useful to show that iDroop can reduce the frequency variance relative to DC and VI. Given the fact that $\|\hat{T}_{\omega\mathrm{dn}, \mathrm{VI}}\|^2_{\mathcal{H}_2}$ is infinite, the question indeed lies in whether we can find a set of values for parameters $\delta$ and $\nu$ that ensure $\|\hat{T}_{\omega\mathrm{dn}, \mathrm{iDroop}}\|^2_{\mathcal{H}_2} \leq \|\hat{T}_{\omega\mathrm{dn}, \mathrm{DC}}\|^2_{\mathcal{H}_2}$. Fortunately, we can not only find such a set but also the optimal setting for \eqref{eq:h2-idroop}. The following three lemmas set the foundation of this important result which is given as Theorem 7.

\begin{lem}[Limit of $\|\hat{T}_{\omega\mathrm{dn}, \mathrm{iDroop}}\|^2_{\mathcal{H}_2}$]\label{lem:h2lim}
Let Assumptions~\ref{ass:proportion} and \ref{ass:noise} hold. If $\delta \to \infty$, then $\|\hat{T}_{\omega\mathrm{dn}, \mathrm{iDroop}}\|^2_{\mathcal{H}_2} = \|\hat{T}_{\omega\mathrm{dn}, \mathrm{DC}}\|^2_{\mathcal{H}_2}$.
\end{lem}
\begin{proof}
The limit of \eqref{eq:h2-idroop} as $\delta \to \infty$ can be computed as
\[
   \underset{\delta \to \infty}{\lim} \|\hat{T}_{\omega\mathrm{dn}, \mathrm{iDroop}}\|^2_{\mathcal{H}_2}
    \!=\! \sum_{k=1}^n \Gamma_{kk} {\dfrac{\kappa_\mathrm{p}^2 + r_\mathrm{r}^{-2} \kappa_\omega^2  }{2m\check{d}}} \!=\! \|\hat{T}_{\omega\mathrm{dn}, \mathrm{DC}}\|^2_{\mathcal{H}_2}\,,
\]
where the second equality follows from \eqref{eq:noise-DC}.
\end{proof}

Lemma~\ref{lem:h2lim} shows that $\|\hat{T}_{\omega\mathrm{dn}, \mathrm{iDroop}}\|_{\mathcal{H}_2}^2$ asymptotically converges to $\|\hat{T}_{\omega\mathrm{dn}, \mathrm{DC}}\|_{\mathcal{H}_2}^2$ as $\delta \to \infty$. 
The next lemma shows that this convergence is monotonically from either above or below depending on the value of the parameter $\nu$.

\begin{lem}[$\nu$-dependent monotonicity of $\|\hat{T}_{\omega\mathrm{dn}, \mathrm{ iDroop}}\|_{\mathcal{H}_2}^2$ with respect to $\delta$ ]\label{lem:mono-alp}
Let Assumptions~\ref{ass:proportion} and \ref{ass:noise} hold. Define
\begin{equation}
	\alpha_1 (\nu)
     := \dfrac{- \check{d}\kappa_\omega^2   \nu^2 + \left(\kappa_\mathrm{p}^2 + r_\mathrm{r}^{-2}\kappa_\omega^2 \right)\nu + d r_\mathrm{r}^{-2} \kappa_\omega^2  - r_\mathrm{r}^{-1}\kappa_\mathrm{p}^2}{d+\nu} \nonumber\,.\label{eq:alpha1}
\end{equation}
Then
\begin{itemize}
\item
$\|\hat{T}_{\omega\mathrm{dn}, \mathrm{iDroop}}\|_{\mathcal{H}_2}^2$ is a monotonically increasing or decreasing function of $\delta > 0$ if and only if $\alpha_1 (\nu)$ is positive or negative, respectively. 
\item 
$\|\hat{T}_{\omega\mathrm{dn}, \mathrm{iDroop}}\|_{\mathcal{H}_2}^2$ is independent of $\delta>0$ if and only if $\alpha_1 (\nu)$ is zero.
\end{itemize}
\end{lem}

\begin{proof}
Provided that $\|\hat{T}_{\omega\mathrm{dn}, \mathrm{iDroop}}\|_{\mathcal{H}_2}^2$ is a function of $\delta$ and $\nu$, in what follows we denote it by $\Pi(\delta, \nu)$. To make it clear how  $\Pi(\delta, \nu)$ changes with $\delta$, we firstly put it into the equivalent form of
\begin{equation*}
\Pi(\delta, \nu) = \sum_{k=1}^n \Gamma_{kk} \left[\dfrac{\alpha_1(\nu) \delta^2}{\alpha_2 \delta^2+\alpha_3(\nu) \delta +\alpha_4(\nu, \lambda_k)} + \alpha_5(\nu) \right]\
    \end{equation*}
with
\begin{subequations} \label{eq:alpha}
    \begin{align*}
        &\alpha_1(\nu) := \dfrac{- \check{d}\kappa_\omega^2   \nu^2 + \left(\kappa_\mathrm{p}^2 + r_\mathrm{r}^{-2}\kappa_\omega^2 \right)\nu + d r_\mathrm{r}^{-2} \kappa_\omega^2  - r_\mathrm{r}^{-1}\kappa_\mathrm{p}^2}{d+\nu}\;,\\
        &\alpha_2 := 2m\check{d}\;,\qquad\;
        \alpha_3(\nu) := 2(d+\nu)\check{d}\;,   \\
        &\alpha_4(\nu, \lambda_k) := 2(d+\nu)\lambda_k\;,\qquad
        \alpha_5(\nu) := \dfrac{\kappa_\mathrm{p}^2+\nu^2 \kappa_\omega^2}{2m(d+\nu)}\;. 
    \end{align*}    
\end{subequations}
We then take the partial derivative of $\Pi(\delta, \nu)$ with respect to $\delta$ as
\begin{align*}
    \partial_{\delta} \Pi(\delta, \nu)
    =\! \alpha_1(\nu) \sum_{k=1}^n \Gamma_{kk} \!\left[ \dfrac{ \alpha_3(\nu) \delta^2 +2 \alpha_4(\nu, \lambda_k) \delta}{(\alpha_2 \delta^2+\alpha_3(\nu) \delta +\alpha_4(\nu, \lambda_k))^2}\right].
\end{align*}

Since $m > 0$, $d > 0$, $\nu > 0$, and $r_\mathrm{r}^{-1} > 0$, $\alpha_2$ and $\alpha_3(\nu)$ are positive. Also, given that all the eigenvalues of the scaled Laplacian matrix $L_\mathrm{F}$ are non-negative, $\alpha_4(\nu, \lambda_k)$ must be non-negative. Thus, $\forall \delta > 0$, $( \alpha_3(\nu) \delta^2 +2 \alpha_4(\nu, \lambda_k) \delta)/(\alpha_2 \delta^2+\alpha_3(\nu) \delta +\alpha_4(\nu, \lambda_k))^2 > 0$. 

Recall from the proof of Corollary~\ref{cor:optimal-h2-dc} that $\Gamma_{kk} > 0$, $\forall  k \in \{1,\dots, n\}$. Therefore, $\forall \delta > 0$, $\mathrm{sign} \left( \partial_{\delta} \Pi(\delta, \nu) \right) = \mathrm{sign} \left( \alpha_1(\nu) \right)$. 
\end{proof}

By Lemma~\ref{lem:mono-alp}, for a given $\nu$, if $\alpha_1(\nu) < 0$, then $\|\hat{T}_{\omega\mathrm{dn}, \mathrm{iDroop}}\|^2_{\mathcal{H}_2}$ always decreases as $\delta$ increases. However, according to Lemma~\ref{lem:h2lim}, even if $\delta \to \infty$, we can only obtain $\|\hat{T}_{\omega\mathrm{dn}, \mathrm{iDroop}}\|^2_{\mathcal{H}_2} = \|\hat{T}_{\omega\mathrm{dn}, \mathrm{DC}}\|^2_{\mathcal{H}_2}$. Similarly, if $\alpha_1 (\nu) = 0$, then $\|\hat{T}_{\omega\mathrm{dn}, \mathrm{iDroop}}\|^2_{\mathcal{H}_2}$ keeps constant as $\delta$ increases, which means whatever $\delta$ is we will always obtain $\|\hat{T}_{\omega\mathrm{dn}, \mathrm{iDroop}}\|^2_{\mathcal{H}_2} = \|\hat{T}_{\omega\mathrm{dn}, \mathrm{DC}}\|^2_{\mathcal{H}_2}$. Therefore, iDroop cannot outperform DC when $\alpha_1 (\nu) \le 0$. To put it another way, Lemmas~\ref{lem:h2lim} and \ref{lem:mono-alp} imply that in order to improve the frequency variance through iDroop, one needs to set $\nu$ such that $\alpha_1(\nu) > 0$ and $\delta$ as small as practically possible. The following lemma characterizes the minimizer $\nu^{\star}$ of $\|\hat{T}_{\omega\mathrm{dn}, \mathrm{iDroop}}\|_{\mathcal{H}_2}^2$ when $\delta = 0$. 

\begin{lem}[Minimizer $\nu^{\star}$ of $\|\hat{T}_{\omega\mathrm{dn}, \mathrm{iDroop}}\|_{\mathcal{H}_2}^2$ when $\delta = 0$]\label{lem:vstar}
Let Assumptions~\ref{ass:proportion} and \ref{ass:noise} hold. Then 
\begin{equation}\label{eq:vstar}
  \nu^{\star}\!\!:=\! \argmin_{\delta = 0, \nu > 0} \!\|\hat{T}_{\omega\mathrm{dn}, \mathrm{iDroop}}\|_{\mathcal{H}_2}^2 \!\!=\! -d +\! \sqrt{d^2 + (\kappa_{\mathrm{p}}/\kappa_{\omega})^2}\, .
\end{equation}
\end{lem}

\begin{proof}
Recall from the proof of Lemma~\ref{lem:mono-alp} that $\|\hat{T}_{\omega\mathrm{dn}, \mathrm{iDroop}}\|_{\mathcal{H}_2}^2 = \Pi(\delta, \nu)$. Then we have
\begin{equation*}
\Pi(0, \nu) = \dfrac{\kappa_\mathrm{p}^2+\nu^2 \kappa_\omega^2}{2m(d + \nu)} \sum_{k=1}^n \Gamma_{kk}\;,
\end{equation*}
whose derivative with respect to $\nu$ is given by
\begin{equation}
        \Pi'(0, \nu) = \dfrac{\kappa_\omega^2\nu^2 + 2 d\kappa_\omega^2 \nu - \kappa_\mathrm{p}^2}{2m (d+\nu)^2} \sum_{k=1}^n \Gamma_{kk}\;. \label{eq:g'(0,nu)}
\end{equation}
Note that \eqref{eq:g'(0,nu)} and \eqref{eq:h2-dc-partial} are in the same form. Thus, $\nu^{\star}$ is determined in the same way as in the proof of Corollary~\ref{cor:optimal-h2-dc}.
\end{proof}

We are now ready to prove the next theorem.

\begin{thm}[$\|\hat{T}_{\omega\mathrm{dn}, \mathrm{iDroop}}\|_{\mathcal{H}_2}^2$ optimal tuning]\label{thm:h2-improves}
Let Assumptions~\ref{ass:proportion} and \ref{ass:noise} hold. Define $\nu^{\star}$ as in \eqref{eq:vstar}. Then 
\begin{itemize}
\item whenever $(\kappa_\mathrm{p}/\kappa_\omega)^2 \neq 2r_\mathrm{r}^{-1}d + r_\mathrm{r}^{-2}$, for any $\delta > 0$ and $\nu$ such that
\begin{equation}\label{eq:condition}
\nu \in [\nu^{\star},r_\mathrm{r}^{-1}) \quad \text{or}\quad \nu \in (r_\mathrm{r}^{-1},\nu^\star]\;,
\end{equation}
iDroop outperforms DC in terms of frequency variance, i.e.,
\begin{equation*}
	\|\hat{T}_{\omega\mathrm{dn}, \mathrm{iDroop}}\|^2_{\mathcal{H}_2}<\|\hat{T}_{\omega\mathrm{dn}, \mathrm{DC}}\|^2_{\mathcal{H}_2}\;.
\end{equation*}
Moreover, the global minimum of $\|\hat{T}_{\omega\mathrm{dn}, \mathrm{iDroop}}\|^2_{\mathcal{H}_2}$ is obtained by setting $\delta \to 0$ and $\nu \to \nu^{\star}$.
\item if $(\kappa_\mathrm{p}/\kappa_\omega)^2 = 2r_\mathrm{r}^{-1}d + r_\mathrm{r}^{-2}$, then for any $\delta > 0$, by setting $\nu \to \nu^{\star} = r_\mathrm{r}^{-1}$, iDroop matches DC in terms of frequency variance, i.e.,
\begin{equation*}
	\|\hat{T}_{\omega\mathrm{dn}, \mathrm{iDroop}}\|^2_{\mathcal{H}_2}=\|\hat{T}_{\omega\mathrm{dn}, \mathrm{DC}}\|^2_{\mathcal{H}_2}\;.
\end{equation*}
\end{itemize}
\end{thm}
\begin{proof}
As discussed before,  to guarantee $\|\hat{T}_{\omega\mathrm{dn}, \mathrm{iDroop}}\|^2_{\mathcal{H}_2} < \|\hat{T}_{\omega\mathrm{dn}, \mathrm{DC}}\|^2_{\mathcal{H}_2}$, one requires to set $\nu$ such that $\alpha_1 (\nu) > 0$. In this case, $\|\hat{T}_{\omega\mathrm{dn}, \mathrm{iDroop}}\|^2_{\mathcal{H}_2}$ always increases as $\delta$ increases, so choosing $\delta$ arbitrarily small is optimal for any fixed $\nu$.

We now look for the values of $\nu$ that satisfy the requirement  $\alpha_1 (\nu) > 0$. Since the denominator of $\alpha_1 (\nu) $ is always positive, the sign of $\alpha_1 (\nu)$ only depends on its numerator. Denote the numerator of $\alpha_1 (\nu)$ as $N_{\alpha_1}(\nu)$. Clearly, $N_{\alpha_1}(\nu)$ is a univariate quadratic function in $\nu$, whose roots are: $\nu_1 = r_\mathrm{r}^{-1}$ and $\nu_2 =  \left[(\kappa_\mathrm{p}/\kappa_\omega)^2 - r_\mathrm{r}^{-1} d \right]/\check{d}$.
Provided that the highest order coefficient of $N_{\alpha_1}(\nu)$ is negative, the graph of $N_{\alpha_1}(\nu)$ is a parabola that opens downwards. Therefore, if $\nu_1 < \nu_2$, then $\nu \in (\nu_1, \nu_2)$ guarantees $\alpha_1 (\nu) > 0$; if $\nu_1 > \nu_2$, then $\nu \in (\nu_2, \nu_1)\cap(0, \infty)$ guarantees $\alpha_1 (\nu) > 0$. Notably, if $\nu_1 = \nu_2$, there exists no feasible points of $\nu$ to make $\alpha_1 (\nu) > 0$. 

The condition $\nu_1 = \nu_2$ happens only if $(\kappa_\mathrm{p}/\kappa_\omega)^2 = 2r_\mathrm{r}^{-1} d + r_\mathrm{r}^{-2}$, from which it follows that $\nu^{\star} = r_\mathrm{r}^{-1} = \nu_1 = \nu_2$. Then $\alpha_1 (\nu^{\star}) = \alpha_1 (r_\mathrm{r}^{-1}) = 0$. Therefore, by setting $\nu \to \nu^{\star} = r_\mathrm{r}^{-1}$, we get $\|\hat{T}_{\omega\mathrm{dn}, \mathrm{iDroop}}\|^2_{\mathcal{H}_2}=\|\hat{T}_{\omega\mathrm{dn}, \mathrm{DC}}\|^2_{\mathcal{H}_2}$. This concludes the proof of the second part. 

We now focus on the case where the set $S=(\nu_1, \nu_2)\cup\{(\nu_2, \nu_1)\cap(0, \infty)\}$
is nonempty. Recall from the proof of Lemma~\ref{lem:mono-alp} that $\|\hat{T}_{\omega\mathrm{dn}, \mathrm{iDroop}}\|_{\mathcal{H}_2} = \Pi(\delta, \nu)$. For any fixed $\nu \in S$, it holds that $\alpha_1 (\nu) > 0$ and thus $\Pi(\delta, \nu) > \Pi(0, \nu)$ for any $\delta >0$. Recall from the proof of Lemma~\ref{lem:vstar} that $\nu^{\star}$ is the minimizer of $\Pi(0, \nu)$. Hence, $(0, \nu^{\star})$ globally minimizes $\Pi(\delta, \nu)$ as long as $\nu^{\star}\in S$. In fact, we will show next that $\nu^{\star}$ is always within $S$ whenever $S\neq\emptyset$.

Firstly we consider the case when $\nu_1 < \nu_2$, which implies that $(\kappa_\mathrm{p}/\kappa_\omega)^2 > 2r_\mathrm{r}^{-1} d + r_\mathrm{r}^{-2}$. Then we have $\nu^{\star} >-d+\sqrt{d^2 + 2r_\mathrm{r}^{-1}d + r_\mathrm{r}^{-2}} = r_\mathrm{r}^{-1} = \nu_1$.
We also want to show $\nu^{\star} < \nu_2$ which holds if and only if
\begin{align*}
	\sqrt{d^2 + (\kappa_\mathrm{p}/\kappa_\omega)^2} \!<\! \dfrac{(\kappa_\mathrm{p}/\kappa_\omega)^2 - r_\mathrm{r}^{-1} d }{\check{d}} + d = \dfrac{(\kappa_\mathrm{p}/\kappa_\omega)^2 + d^2 }{\check{d}}
\end{align*}
which is equivalent to $1 < \sqrt{d^2 + (\kappa_\mathrm{p}/\kappa_\omega)^2}/\check{d}$. 
This always holds since $(\kappa_\mathrm{p}/\kappa_\omega)^2 > 2r_\mathrm{r}^{-1} d + r_\mathrm{r}^{-2}$. Thus, $\nu_1 < \nu^{\star} < \nu_2$.
Similarly, we can prove that in the case when $\nu_1 > \nu_2$, $\nu_2 < \nu^{\star} < \nu_1$ holds and thus $\nu^{\star} \in (\nu_2, \nu_1)\cap(0, \infty)$.
It follows that $(0, \nu^{\star})$ is the global minimizer of $\Pi(\delta, \nu)$.

Finally, by Lemma~\ref{lem:h2lim}, $\|\hat{T}_{\omega\mathrm{dn}, \mathrm{DC}}\|^{2}_{\mathcal{H}_2} = \Pi(\infty, \nu)$. The condition \eqref{eq:condition} actually guarantees $\nu \in S$ and thus $\alpha_1(\nu) > 0$. Then, by Lemma~\ref{lem:mono-alp}, we have $\|\hat{T}_{\omega\mathrm{dn}, \mathrm{DC}}\|^{2}_{\mathcal{H}_2} = \Pi(\infty, \nu) > \Pi(\delta, \nu)$.
This concludes the proof of the first part.
\end{proof} 

Theorem \ref{thm:h2-improves} shows that, to optimally improve the frequency variance, iDroop needs to first set $\delta$ arbitrarily close to zero. Interestingly, this implies that the transfer function $\hat{c}_\mathrm{o}(s)\approx-\nu$ except for $\hat{c}_\mathrm{o}(0)=-r_\mathrm{r}^{-1}$. In other words, iDroop uses its first-order lead/lag property to effectively decouple the dc gain $\hat{c}_\mathrm{o}(0)$ from the gain at all the other frequencies such that $\hat{c}_\mathrm{o}(\boldsymbol{j\omega})\approx -\nu$. This decouple is particularly easy to understand in two special regimes: (i) If $\kappa_\mathrm{p}\ll \kappa_\omega$, the system is dominated by measurement noise and therefore $\nu^{\star} \approx 0 <r_\mathrm{r}^{-1}$ which makes iDroop a lag compensator. Thus, by using lag compensation (setting $\nu<r_\mathrm{r}^{-1}$) iDroop can attenuate frequency noise; 
(ii) If $\kappa_\mathrm{p} \gg \kappa_\omega$, the system is dominated by power fluctuations and therefore $\nu^\star \approx \kappa_\mathrm{p}/\kappa_\omega >r_\mathrm{r}^{-1}$ which makes iDroop a lead compensator. Thus, by using lead compensation (setting $\nu>r_\mathrm{r}^{-1}$) iDroop can mitigate power fluctuations. 

\subsection{Synchronization Cost}
Theorem~\ref{thm:bound-cost} implies that the bounds on the synchronization cost of $\hat{T}_{\omega \mathrm{p, iDroop}}$ are closely related to $\| \hat{h}_{\mathrm{u},k,\mathrm{iDroop}} \|_{\mathcal{H}_2}^2$. If we can find a tuning that forces $\| \hat{h}_{\mathrm{u},k,\mathrm{iDroop}}\|_{\mathcal{H}_2}^2$ to be zero, then both lower and upper bounds on the synchronization cost converge to zero. Then, the zero synchronization cost is achieved naturally. The next theorem addresses this problem.

\begin{thm}[Zero synchronization cost tuning of iDroop] 
\label{th:zero-syn-idroop}
Let Assumptions~\ref{ass:proportion} and \ref{ass:step} hold. Then a zero synchronization cost of the system $\hat{T}_{\omega \mathrm{p, iDroop}}$, i.e., $\|\tilde{\omega}_\mathrm{iDroop}\|_2^2 = 0$, can be achieved by setting $\delta \to 0$ and $\nu \to \infty$.
\end{thm}
\begin{proof}
Since the key is to show that $\| \hat{h}_{\mathrm{u},k,\mathrm{iDroop}} \|_{\mathcal{H}_2}^2 \to 0$ as $\delta \to 0$ and $\nu \to \infty$, 
we can use Lemma~\ref{lm:h2-4th}. Applying \eqref{eq:go-sw-tb} and \eqref{eq:co-idroop} to \eqref{eq:hp-s} shows $\hat{h}_{\mathrm{u},k,\mathrm{iDroop}}(s)=\hat{h}_{\mathrm{p},k+1,\mathrm{T,iDroop}}(s)/s$ is a transfer function with
\begin{subequations}
\begin{align*}
    a_0 =& \frac{\lambda_{k+1}\delta}{m \tau}\;,
    \qquad b_0 = \frac{\delta}{m \tau}\;,\\
    a_1 =& \frac{ \delta(\check{d}+ r_{\mathrm{t}}^{-1}+\lambda_{k+1}\tau) + \lambda_{k+1}}{m \tau}\;,
    \qquad b_1 = \frac{\delta \tau + 1}{m \tau}\;,\\
    a_2 =& \frac{\delta(m+ \check{d} \tau )  + d + r_{\mathrm{t}}^{-1} + \lambda_{k+1} \tau + \nu }{m \tau},\qquad b_2 =\! \frac{1}{m}\;,\\
    a_3 =& \frac{m \delta \tau + m + d \tau + \nu \tau}{m \tau},
    \qquad b_3 = 0\;,\qquad b_4 = 0\;.
\end{align*}
\end{subequations}
Considering that $a_0\to0$ and $b_0\to0$ as $\delta \to 0$ and $\nu \to \infty$,
we can employ the $\mathcal{H}_2$ norm computation formula for the third-order transfer function in Remark~\ref{rem:h2-3rd}. Then
\begin{align*}
\underset{\delta \to 0, \nu \to \infty}{\lim} \!\!\|\hat{h}_{\mathrm{u},k,\mathrm{iDroop}}\|_{\mathcal{H}_2}^2
\!\!
\!=&\!\underset{\delta \to 0, \nu \to \infty}{\lim} \!\frac{\frac{\nu}{m} \!\left(\frac{1}{m \tau}\right)^2\!+\!\frac{\lambda_{k+1} }{m \tau} \!\left(\frac{1}{m}\right)^2}{2 \frac{\lambda_{k+1} }{m \tau} (\frac{\nu}{m \tau} \frac{\nu}{m}\!-\! \frac{\lambda_{k+1} }{m \tau})}
\!=\!0\,.
\end{align*}
Thus by Theorem~\ref{thm:bound-cost}, $\underline{\|\tilde{\omega}_\mathrm{iDroop}\|_2^2}=\overline{\|\tilde{\omega}_\mathrm{iDroop}\|_2^2}=0$, which forces $\|\tilde{\omega}_\mathrm{iDroop}\|_2^2 = 0$.
\end{proof}

Theorem~\ref{th:zero-syn-idroop} shows that unlike DC and VI that require changes on $r_\mathrm{r}^{-1}$ to arbitrarily reduce the synchronization cost, iDroop can achieve zero synchronization cost without affecting the steady-state effort share. Naturally, $\delta\approx0$ may lead to slow response and $\nu\rightarrow\infty$ may hinder robustness. Thus this result should be appreciated from the viewpoint of the additional tuning flexibility that iDroop provides.

\subsection{Nadir}
Finally, we show that with $\delta$ and $\nu$ tuned appropriately, iDroop enables the system frequency of $\hat{T}_{\omega \mathrm{p, iDroop}}$ to evolve as a first-order response to step power disturbances, which effectively makes Nadir disappear. The following theorem summarizes this idea.
\begin{thm}[Nadir elimination with iDroop]\label{thm:no nadir}
Let Assumptions~\ref{ass:proportion} and \ref{ass:step} hold. By setting $\delta = \tau^{-1}$ and $\nu = r_\mathrm{r}^{-1} + r_\mathrm{t}^{-1}$, Nadir \eqref{eq:Nadir} of $\hat{T}_{\omega \mathrm{p, iDroop}}$ disappears.
\end{thm}
\begin{proof}
The system frequency of $\hat{T}_{\omega \mathrm{p, iDroop}}$ is given by \cite{p2017ccc}
\begin{equation}\label{eq:nadir}
	\bar{\omega}_\mathrm{iDroop}(t) =  \dfrac{\sum_{i=1}^n u_{0,i} }{ \sum_{i=1}^n f_i } p_\mathrm{u,iDroop}(t)\;,
\end{equation}
where $p_\mathrm{u,iDroop}(t)$ is the unit-step response of $\hat{h}_{{\mathrm{p},1},\mathrm{T, iDroop}}(s)$. If we set $\delta = \tau^{-1}$ and $\nu = r_\mathrm{r}^{-1} + r_\mathrm{t}^{-1}$, then \eqref{eq:co-idroop} becomes
\begin{equation}\label{eq:co-idroop-nonadir}
\hat{c}_\mathrm{o}(s) =\frac{ r_\mathrm{t}^{-1} }{\tau s+1}-\left(r_\mathrm{r}^{-1} + r_\mathrm{t}^{-1}\right)\;.
\end{equation}
Applying \eqref{eq:go-sw-tb} and \eqref{eq:co-idroop-nonadir} to \eqref{eq:hp-s} yields
\begin{align*} 
    \hat{h}_{{\mathrm{p},1},\mathrm{T, iDroop}}(s)
     =& \dfrac{1}{m s + \check{d} + r_{\mathrm{t}}^{-1}}\;,
\end{align*}
whose unit-step response $p_\mathrm{u,iDroop}(t)$ is a first-order evolution. Thus, \eqref{eq:nadir} indicates that Nadir of the system frequency disappears.
\end{proof}

%% file: 06-numerical_illustration.tex
In this section, we present simultation results that compare iDroop with DC and VI.
The simulations are performed on the Icelandic Power Network taken from the Power Systems Test Case Archive \cite{iceland}. The dynamic model is built upon the Kron reduced system \cite{Dorfler2013kron} where only the $35$ generator buses are retained. Even though our previous analysis requires the proportionality assumption (Assumption \ref{ass:proportion}), in the simulations, for every bus $i$, the generator inertia coefficient, the turbine time constant, and the turbine droop coefficient are directly obtained from the dataset, i.e., $m_i = m_{\mathrm{d},i}$, $\tau_i = \tau_{\mathrm{d},i}$, and $r_{\mathrm{t},i} = r_{\mathrm{t},\mathrm{d},i}$.\footnotemark[5]\footnotemark[6] In addition, turbine governor deadbands are taken into account such that turbines are only responsive to frequency deviations exceeding \SI{\pm0.036}{\hertz}. Given that the values of generator damping coefficients are not provided by the dataset, we set $d_i = f_i d$ with $d$ being the representative generator damping coefficient and
\begin{equation*}
    f_i := \frac{m_i}{m}
\end{equation*}
being the proportionality parameters, where $m$ is the representative generator inertia defined as the mean of $m_i$'s, i.e.,
\begin{equation*}
    m := \frac{1}{n}\sum_{i=1}^n m_i.
\end{equation*}
We refer to this system without inverter control to 'SW' in the simulations.
\footnotetext[5]{Throughout this section, we use the subscript ${\mathrm{d},i}$ to denote the original parameters of the $i$th generator bus from the dataset.}
\footnotetext[6]{For illustrative purpose only, we reassign a part of the droop $r_{\mathrm{t},\mathrm{d},i}$'s on turbines in the dataset to let there be a deeper Nadir in the system frequency.}

We then add an inverter to each bus $i$, whose control law is either one of DC, VI, and iDroop. The design of controller parameters will be based on the representative generator parameters. Hence, besides $m$ and $d$, we define
\begin{equation*}
    \tau := \frac{1}{n}\sum_{i=1}^{n} \tau_{\mathrm{d},i}\quad\text{and}\quad
    r_{\mathrm{t}} := \frac{\sum_{i=1}^n f_i}{\sum_{i=1}^n r_{\mathrm{t},\mathrm{d},i}^{-1}}.
\end{equation*}
Note that to keep the synchronous frequency unchanged, once inverters are added, we halve the inverse turbine droop $r_{\mathrm{t},i}^{-1}$ and assign the representative inverter droop coefficient $r_{\mathrm{r}}$ a value such that the inverse inverter droop $ r_{\mathrm{r},i}^{-1} := f_i r_{\mathrm{r}}^{-1}$ should exactly compensate this decreased $r_{\mathrm{t}}^{-1}$ in the absence of turbine governor deadbands.
The values of all the representative parameters mentioned above are given in Table~\ref{table:parameter}.


\begin{table}
\renewcommand{\arraystretch}{1.2}
\centering
\small
\begin{tabular}{c|c|c}
\arrayrulecolor[rgb]{.7,.8,.9}
\hline\hline
Parameters & Symbol & Value\\
\hline
generator inertia
& $m$ & \SI{0.0111}{\square\second\per\radian} \\
\hline
generator damping & $d$ & \SI{0.0014}{\second\per\radian}\\
\hline
turbine time const. & $\tau$ & \SI{4.59}{\second} \\
\hline
\multirow{2}{*}{turbine droop
} & \multirow{2}{*}{$r_{\mathrm{t}}$} & \SI{374.49}{\radian\per\second} for SW,\\ &&\SI{748.97}{\radian\per\second} o.w.  \\
\hline
inverter droop
& $r_{\mathrm{r}}$ & \SI{748.97}{\radian\per\second}\\
\hline\hline
\end{tabular}
\caption{Parameters of Representative Generator and Inverter}
\label{table:parameter}
\end{table}

\subsection{Comparison in Step Input Scenario}\label{sub:step-simulation}
Fig.~\ref{fig:step simulation 1} shows how different controllers perform when the system suffers from a step drop of $-0.3$ p.u. in power injection at bus number $2$ at time $t = \SI{1}{\second}$. As for the representative inverter, we turn $\delta = \tau^{-1} = \SI{0.218}{\per\second}$ and $\nu = r_\mathrm{r}^{-1} + r_\mathrm{t}^{-1} = \SI{0.004}{\second\per\radian}$ in iDroop such that Nadir of the system frequency disappears as suggested by Theorem~\ref{thm:no nadir} and we tune $m_\mathrm{v} = \SI{0.022}{\square\second\per\radian}$ in VI such that the system frequency is critically damped.\footnotemark[7] The inverter parameters on each bus $i$ are defined as follows: $\delta_i := \delta$, $\nu_i := f_i \nu$, and  $m_{\mathrm{v},i} = f_i m_{\mathrm{v}}$.

\footnotetext[7]{In the rest of this section, we keep tuning $m_\mathrm{v} = \SI{0.022}{\square\second\per\radian}$.}

The results are shown in Fig.~\ref{fig:step simulation 1}. One observation is that all three controllers lead to the same synchronous frequency as predicted by Corollaries~\ref{lem:syn-fre-dc} and~\ref{lem:syn-fre-idroop}. Another observation is that although both of VI and iDroop succeed in eliminating Nadir of the system frequency --which is better than what DC does-- the system synchronizes with much faster rate and lower cost under iDroop than VI. Interestingly, the synchronization cost under VI is even slightly higher than that under DC, which indicates that the benefit of eliminating Nadir through increasing $m_\mathrm{v}$ in VI is significantly diluted by the obvious sluggishness introduced to the synchronization process in the meanwhile. Finally, we highlight the huge control effort required by VI when compared with DC and iDroop.

\begin{figure}[t!]
\centering
\subfigure[Frequency deviations]
{\includegraphics[width=\columnwidth]{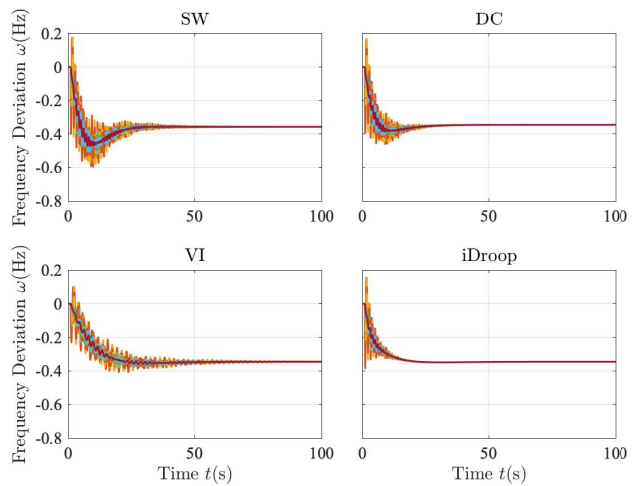}}
\hfil
\subfigure[Control effort]
{\includegraphics[width=\columnwidth]{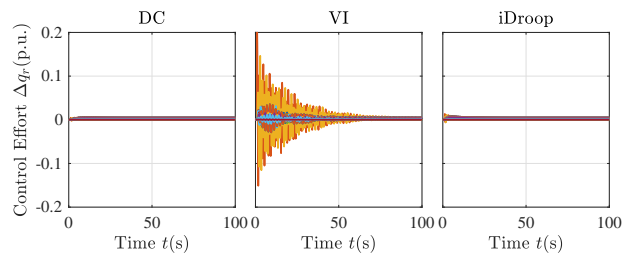}}
\hfil
\subfigure[System frequency and synchronization cost]
{\includegraphics[width=\columnwidth]{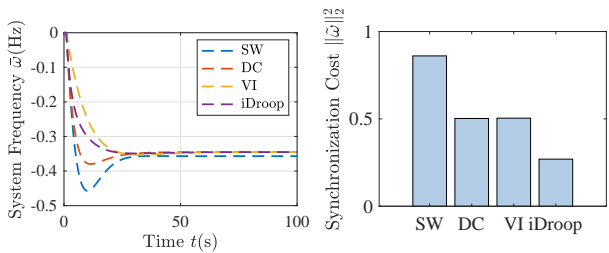}}
\caption{Comparison between controllers when a $-0.3$ p.u. step change in power injection is introduced to bus number $2$.}
\label{fig:step simulation 1}
\end{figure}

\subsection{Comparison in Noise Scenario}
\ifthenelse{\boolean{archive}}{Fig.~\ref{fig:noise simulation pdom} and Fig.~\ref{fig:noise simulation ndom} show how different controllers perform when the system encounters power fluctuations and measurement noise. Fig.~\ref{fig:noise simulation pdom} focuses on the case dominated by power fluctuations where $\kappa_\mathrm{p} = 10^{-4}$ and $\kappa_\omega = 10^{-5}$, while Fig.~\ref{fig:noise simulation ndom} corresponds to the case dominated by measurement noise where $\kappa_\mathrm{p} = 10^{-4}$ and $\kappa_\omega = 10^{-3}$. As required by Theorem~\ref{thm:h2-improves}, we tune $\delta$ to be a small value $\SI{0.1}{\per\second}$ and $\nu$ to be the optimal value $\nu^{\star}$ which is either $\SI{9.9986}{\second\per\radian}$ for $\kappa_\mathrm{p} \gg \kappa_\omega$ or $\SI{0.0986}{\second\per\radian}$ for $\kappa_\mathrm{p} \ll \kappa_\omega$.

 Observe from Fig.~\ref{fig:noise simulation pdom-fre} and Fig.~\ref{fig:noise simulation ndom-fre} that setting $\delta$ small enough and $\nu$ optimally guarantees iDroop has a better performance than DC in terms of noise variance, which actually fits well with Theorem~\ref{thm:h2-improves}. Note that, as expected from Corollary~\ref{thm:noise-VI}, whichever type of noise dominates the system under VI performs badly, therefore whenever noise appears the simulation results of VI will be omitted throughout this section.}
 {Fig.~\ref{fig:noise simulation pdom} shows how different controllers perform when the system encounters power fluctuations and measurement noise. Since in reality power fluctuations are larger than measurement noise, we focus on the case dominated by power fluctuations, where $\kappa_\mathrm{p} = 10^{-4}$ and $\kappa_\omega = 10^{-5}$. As required by Theorem~\ref{thm:h2-improves}, we tune $\delta$ to be a small value $\SI{0.1}{\per\second}$ and $\nu$ to be the optimal value $\nu^{\star}$ which is $\SI{9.9986}{\second\per\radian}$ here.

 Observe from Fig.~\ref{fig:noise simulation pdom-fre} that setting $\delta$ small enough and $\nu=\nu^\star$ ensures that iDroop has a better performance than DC in terms of frequency variance, as expected by Theorem~\ref{thm:h2-improves}. Note that, since by Corollary~\ref{thm:noise-VI}, VI performs badly, we do not evaluate VI in the presence of stochastic disturbances.}

\begin{figure}[t!]
\centering
\subfigure[Frequency deviations]
{\includegraphics[width=\columnwidth]{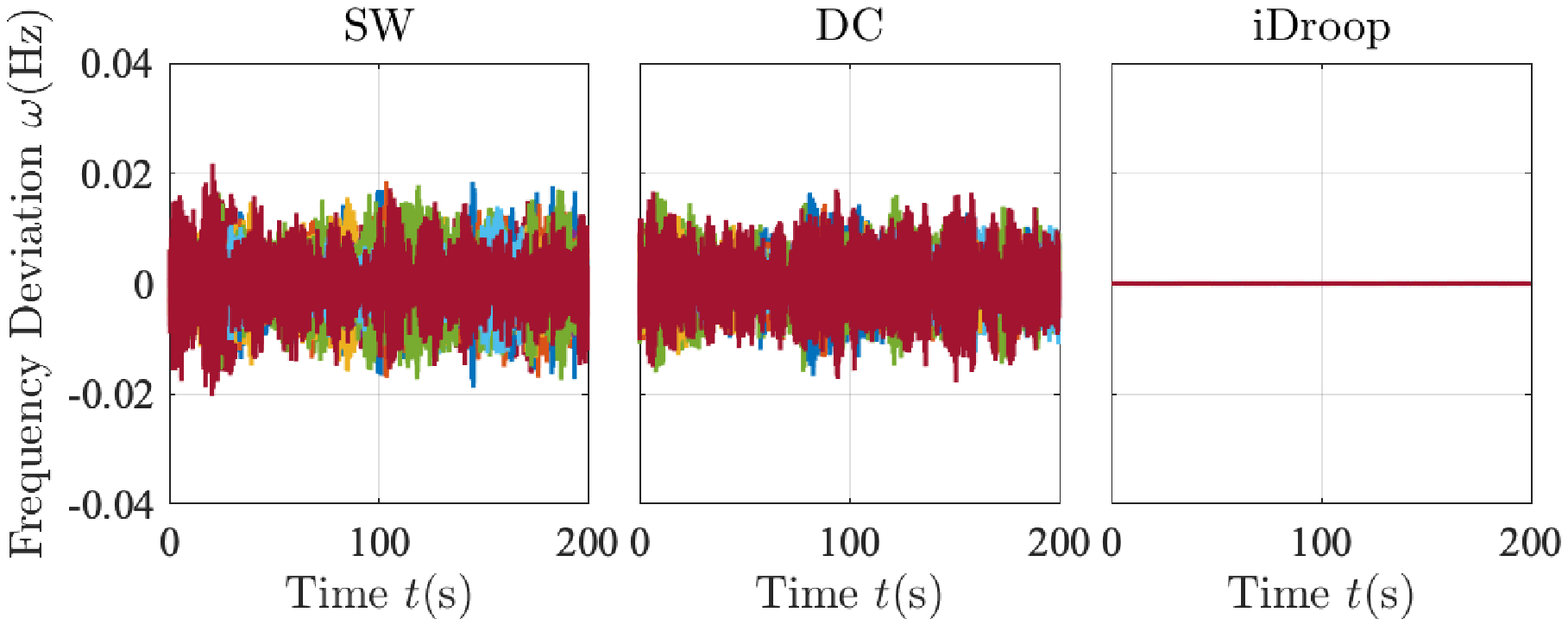}\label{fig:noise simulation pdom-fre}}
\hfil
\subfigure[Control effort]
{\includegraphics[width=\columnwidth]{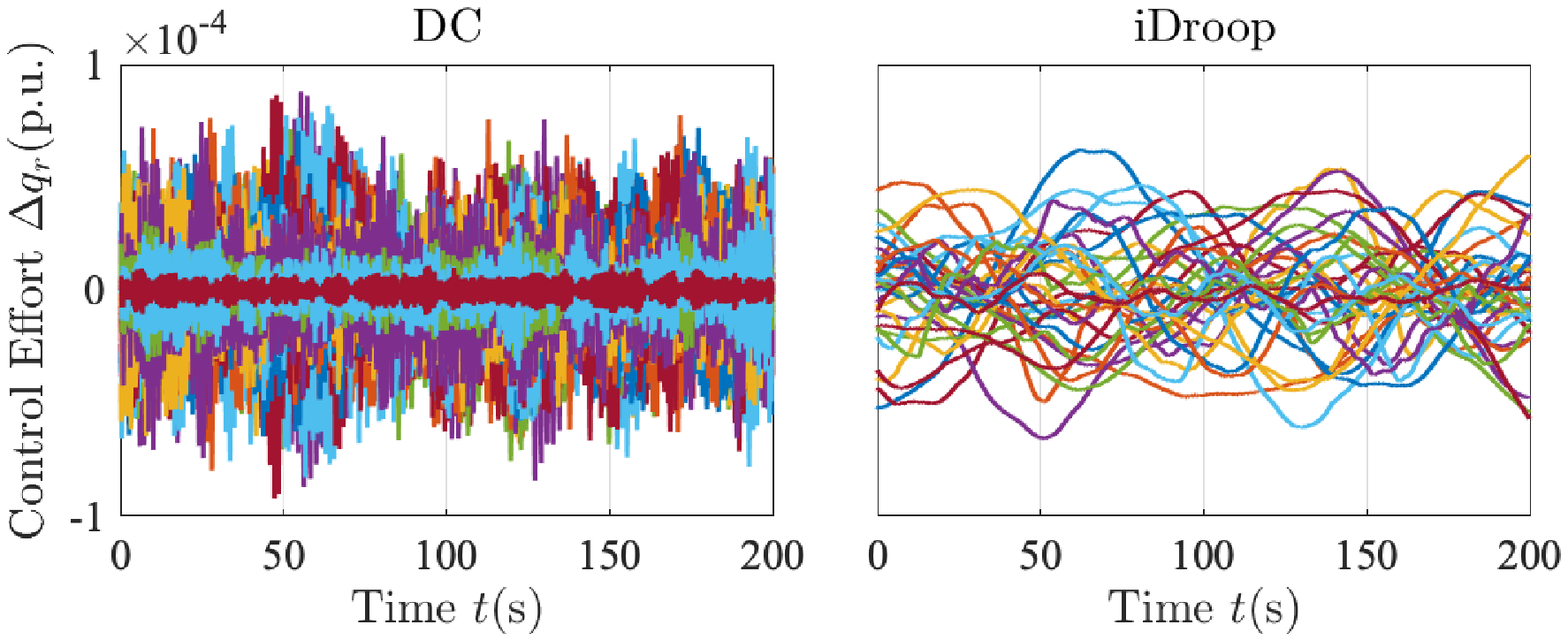}}

\caption{Comparison between controllers when power fluctuations and measurement noise are introduced with $\kappa_\mathrm{p} = 10^{-4}$ and $\kappa_\omega = 10^{-5}$.}
\label{fig:noise simulation pdom}
\end{figure}

\ifthenelse{\boolean{archive}}{
\begin{figure}[t!]
\centering
\subfigure[Frequency deviations]
{\includegraphics[width=\columnwidth]{frequency_noise_ndom_sde_db.eps}\label{fig:noise simulation ndom-fre}}
\hfil
\subfigure[Control effort]
{\includegraphics[width=\columnwidth]{control_noise_ndom_sde_db.eps}}
\hfil
\subfigure[Empirical PDF of frequency deviations]
{\includegraphics[width=\columnwidth]{pdf_noise_ndom_sde_db.eps}}
\caption{Comparison between controllers when power fluctuations and measurement noises are introduced with $\kappa_\mathrm{p} = 10^{-4}$ and $\kappa_\omega = 10^{-3}$.}
\label{fig:noise simulation ndom}
\end{figure}}

\subsection{Tuning for Combined Noise and Step Disturbances}\label{sub:comb-simulation}
Although our current study does not contemplate jointly step and stochastic disturbances, we illustrate here that the Nadir eliminated tuning of Theorem \ref{thm:no nadir} for iDroop can perform quite well in more realistic scenarios with combined step and stochastic disturbances.

 In Fig.~\ref{fig:combine simulation pdom np}, we show how different controllers perform when the system is subject to a step drop of $-0.3$p.u. in power injection at bus number $2$ at time $t = \SI{1}{\second}$ as well as power fluctuations and measurement noise. \ifthenelse{\boolean{archive}}{Since in reality usually power disturbances have a higher order of magnitude than measurement noise does}{Again}, we consider the case with $\kappa_\mathrm{p} = 10^{-4}$ and $\kappa_\omega = 10^{-5}$. Here we employ the same inverter parameters setting as in the step input scenario. More precisely, we tune inverter parameters in iDroop on each bus $i$ as follows: $\delta_i := \delta$, $\nu_i := f_i \nu$, where $\delta = \tau^{-1} = \SI{0.218}{\per\second}$ and $\nu = r_\mathrm{r}^{-1} + r_\mathrm{t}^{-1} = \SI{0.004}{\second\per\radian}$.

 Some observations are in order. First, even though the result is not given here, there is no surprise that the system under VI performs badly due to its inability to reject noise.
 Second, the performance of the system under DC and iDroop is similar to the one in the step input scenario except additional noise. Last but not least, a bonus of the Nadir eliminated tuning is that iDroop outperforms DC in frequency variance as well. This can be understood through Theorem~\ref{thm:h2-improves}. Provided that $\kappa_\mathrm{p} \gg \kappa_\omega$, we know from the definition in Lemma~\ref{lem:vstar} that $\nu^{\star} \approx \kappa_{\mathrm{p}}/\kappa_{\omega}$. Thus, for realistic values of system parameters, $\nu^{\star} \gg r_\mathrm{r}^{-1}$ always holds. It follows directly that $\nu = r_\mathrm{r}^{-1} + r_\mathrm{t}^{-1} \in (r_\mathrm{r}^{-1},\nu^{\star}]$. By Theorem~\ref{thm:h2-improves}, iDroop performs better than DC in terms of frequency variance. Further, the preceding simulation results suggest that the Nadir eliminated tuning of iDroop designed based on the proportional parameter assumption works relatively well even when parameters are non-proportional.


\begin{figure}[t!]
\centering
\subfigure[Frequency deviations]
{\includegraphics[width=\columnwidth]{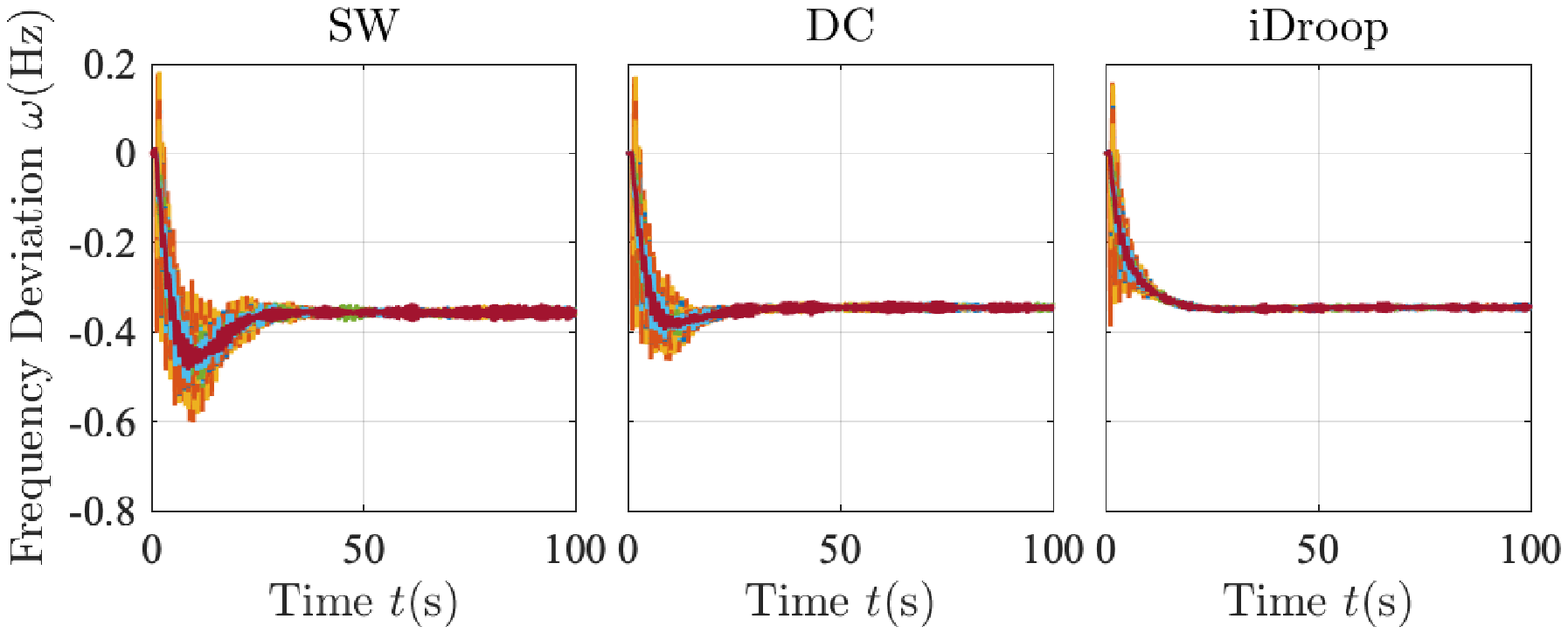}}
\hfil
\subfigure[Empirical PDF of frequency deviations and system frequency]
{\includegraphics[width=\columnwidth]{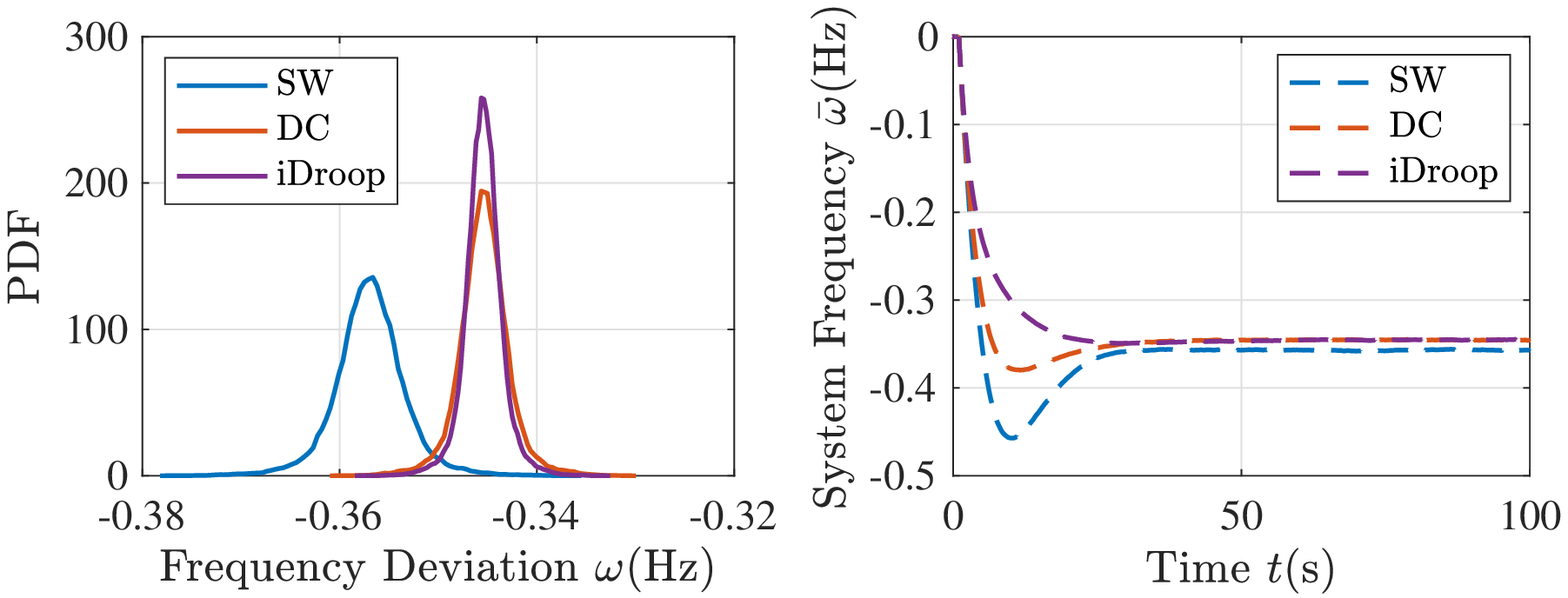}}
\caption{Comparison between controllers when a $-0.3$ p.u. step change in power injection is introduced to bus number $2$ and power fluctuations and measurement noise are introduced with $\kappa_\mathrm{p} = 10^{-4}$ and $\kappa_\omega = 10^{-5}$.}
\label{fig:combine simulation pdom np}
\end{figure}


%% file: 07-conclusion.tex
This paper studies the effect of grid-connected inverter-based control on the power system performance. When it comes to the existing two common control strategies, we show that DC cannot decouple the dynamic performance improvement from the steady-state effort share and VI can introduce unbounded frequency variance. Therefore, we propose a new control strategy named iDroop, which is able to enhance the dynamic performance and preserve the steady-state effort share at the same time. We show that iDroop can be tuned to achieve strong noise rejection, zero synchronization cost, and frequency Nadir elimination when the system parameters satisfy the proportionality assumption. We illustrate numerically that the Nadir eliminated tuning designed based on the proportional parameters assumption strikes a good trade-off among various performance metrics even if parameters are non-proportional.

%% file: main.bbl
\begin{thebibliography}{10}
\providecommand{\url}[1]{#1}
\csname url@samestyle\endcsname
\providecommand{\newblock}{\relax}
\providecommand{\bibinfo}[2]{#2}
\providecommand{\BIBentrySTDinterwordspacing}{\spaceskip=0pt\relax}
\providecommand{\BIBentryALTinterwordstretchfactor}{4}
\providecommand{\BIBentryALTinterwordspacing}{\spaceskip=\fontdimen2\font plus
\BIBentryALTinterwordstretchfactor\fontdimen3\font minus
  \fontdimen4\font\relax}
\providecommand{\BIBforeignlanguage}[2]{{%
\expandafter\ifx\csname l@#1\endcsname\relax
\typeout{** WARNING: IEEEtran.bst: No hyphenation pattern has been}%
\typeout{** loaded for the language `#1'. Using the pattern for}%
\typeout{** the default language instead.}%
\else
\language=\csname l@#1\endcsname
\fi
#2}}
\providecommand{\BIBdecl}{\relax}
\BIBdecl

\bibitem{m2016cdc}
E.~{Mallada}, ``{iDroop}: A dynamic droop controller to decouple power grid's
  steady-state and dynamic performance,'' in \emph{Proc. of IEEE Conference on
  Decision and Control}, Dec. 2016, pp. 4957--4964.

\bibitem{y2017cdc}
Y.~{Jiang}, R.~{Pates}, and E.~{Mallada}, ``Performance tradeoffs of
  dynamically controlled grid-connected inverters in low inertia power
  systems,'' in \emph{Proc. of IEEE Conference on Decision and Control}, Dec.
  2017, pp. 5098--5105.

\bibitem{benjamin2017}
B.~Kroposki, B.~Johnson, Y.~Zhang, V.~Gevorgian, P.~Denholm, B.~Hodge, and
  B.~Hannegan, ``Achieving a 100\% renewable grid: Operating electric power
  systems with extremely high levels of variable renewable energy,'' \emph{IEEE
  Power and Energy Magazine}, vol.~15, no.~2, pp. 61--73, Mar. 2017.

\bibitem{milano2018}
F.~Milano, F.~D\"orfler, G.~Hug, D.~J. Hill, and G.~Verbi{\v c}, ``Foundations
  and challenges of low-inertia systems (invited paper),'' in \emph{Proc. of
  Power Systems Computation Conference}, June 2018, pp. 1--25.

\bibitem{ackermann2017}
T.~Ackermann, T.~Prevost, V.~Vittal, A.~J. Roscoe, J.~Matevosyan, and
  N.~Miller, ``Paving the way: A future without inertia is closer than you
  think,'' \emph{IEEE Power and Energy Magazine}, vol.~15, no.~6, pp. 61--69,
  Nov. 2017.

\bibitem{ulbig2014impact}
A.~Ulbig, T.~S. Borsche, and G.~Andersson, ``Impact of low rotational inertia
  on power system stability and operation,'' in \emph{Proc. of IFAC World
  Congress}, Aug. 2014, pp. 7290--7297.

\bibitem{ahmadyar2018}
A.~S. Ahmadyar, S.~Riaz, G.~Verbi{\v c}, A.~Chapman, and D.~J. Hill, ``A
  framework for assessing renewable integration limits with respect to
  frequency performance,'' \emph{IEEE Transactions on Power Systems}, vol.~33,
  no.~4, pp. 4444--4453, July 2018.

\bibitem{o2014studying}
J.~{O'Sullivan}, A.~{Rogers}, D.~{Flynn}, P.~{Smith}, A.~{Mullane}, and
  M.~{O'Malley}, ``Studying the maximum instantaneous non-synchronous
  generation in an island system---{F}requency stability challenges in
  {I}reland,'' \emph{IEEE Transactions on Power Systems}, vol.~29, no.~6, pp.
  2943--2951, Nov. 2014.

\bibitem{bose2013}
B.~K. Bose, ``Global energy scenario and impact of power electronics in 21st
  century,'' \emph{IEEE Transactions on Industrial Electronics}, vol.~60,
  no.~7, pp. 2638--2651, July 2013.

\bibitem{ofir2018DCvsvVI}
R.~Ofir, U.~Markovic, P.~Aristidou, and G.~Hug, ``Droop vs. virtual inertia:
  Comparison from the perspective of converter operation mode,'' in \emph{Proc.
  of IEEE International Energy Conference}, June 2018, pp. 1--6.

\bibitem{Poolla2019place}
B.~K. {Poolla}, D.~{Gro{\ss}}, and F.~{D{\"o}rfler}, ``Placement and
  implementation of grid-forming and grid-following virtual inertia and fast
  frequency response,'' \emph{IEEE Transactions on Power Systems}, vol.~34,
  no.~4, pp. 3035--3046, July 2019.

\bibitem{Guggilam2018TPS}
S.~S. {Guggilam}, C.~{Zhao}, E.~{Dall'Anese}, Y.~C. {Chen}, and S.~V. {Dhople},
  ``Optimizing {DER} participation in inertial and primary-frequency
  response,'' \emph{IEEE Transactions on Power Systems}, vol.~33, no.~5, pp.
  5194--5205, Sept. 2018.

\bibitem{Markovic2019SE}
U.~{Markovic}, Z.~{Chu}, P.~{Aristidou}, and G.~{Hug}, ``{LQR}-based adaptive
  virtual synchronous machine for power systems with high inverter
  penetration,'' \emph{IEEE Transactions on Sustainable Energy}, vol.~10,
  no.~3, pp. 1501--1512, July 2019.

\bibitem{guo2018cdc}
L.~{Guo}, C.~{Zhao}, and S.~H. {Low}, ``Graph laplacian spectrum and primary
  frequency regulation,'' in \emph{Proc. of IEEE Conference on Decision and
  Control}, Dec. 2018, pp. 158--165.

\bibitem{pm2019preprint}
F.~Paganini and E.~Mallada, ``Global analysis of synchronization performance
  for power systems: bridging the theory-practice gap,'' \emph{arXiv
  preprint:1905.06948}, May 2019.

\bibitem{pagnier2019optimal}
L.~Pagnier and P.~Jacquod, ``Optimal placement of inertia and primary control:
  {A} matrix perturbation theory approach,'' \emph{arXiv preprint: 1906.06922},
  June 2019.

\bibitem{p2017ccc}
F.~Paganini and E.~Mallada, ``Global performance metrics for synchronization of
  heterogeneously rated power systems: The role of machine models and
  inertia,'' in \emph{Proc. of Allerton Conference on Communication, Control,
  and Computing}, Oct. 2017, pp. 324--331.

\bibitem{Brabandere2007dc}
K.~D. Brabandere, B.~Bolsens, J.~V. den Keybus, A.~Woyte, J.~Driesen, and
  R.~Belmans, ``A voltage and frequency droop control method for parallel
  inverters,'' \emph{IEEE Transactions on Power Electronics}, vol.~22, no.~4,
  pp. 1107--1115, July 2007.

\bibitem{beck2007virtual}
H.~Beck and R.~Hesse, ``Virtual synchronous machine,'' in \emph{Proc. of
  International Conference on Electrical Power Quality and Utilisation}, Oct.
  2007, pp. 1--6.

\bibitem{g2015tran}
E.~Tegling, B.~Bamieh, and D.~F. Gayme, ``The price of synchrony: Evaluating
  the resistive losses in synchronizing power networks,'' \emph{IEEE
  Transactions on Control of Network Systems}, vol.~2, no.~3, pp. 254--266,
  Sept. 2015.

\bibitem{Purchala2005dc-flow}
K.~{Purchala}, L.~{Meeus}, D.~{Van Dommelen}, and R.~{Belmans}, ``Usefulness of
  {DC} power flow for active power flow analysis,'' in \emph{Proc. of IEEE
  Power Engineering Society General Meeting}, June 2005, pp. 454--459.

\bibitem{kundur_power_1994}
P.~Kundur, \emph{{Power System Stability and Control}}.\hskip 1em plus 0.5em
  minus 0.4em\relax McGraw-Hill, 1994.

\bibitem{Zhao:2013ts}
C.~Zhao, U.~Topcu, N.~Li, and S.~H. Low, ``Power system dynamics as primal-dual
  algorithm for optimal load control,'' \emph{arXiv preprint:1305.0585}, May
  2013.

\bibitem{Zhao:2014bp}
C.~{Zhao}, U.~{Topcu}, N.~{Li}, and S.~{Low}, ``Design and stability of
  load-side primary frequency control in power systems,'' \emph{IEEE
  Transactions on Automatic Control}, vol.~59, no.~5, pp. 1177--1189, May 2014.

\bibitem{Li:2016tcns}
N.~Li, C.~Zhao, and L.~Chen, ``Connecting automatic generation control and
  economic dispatch from an optimization view,'' \emph{IEEE Transactions on
  Control of Network Systems}, vol.~3, no.~3, pp. 254--264, Sept. 2016.

\bibitem{mallada2017optimal}
E.~{Mallada}, C.~{Zhao}, and S.~{Low}, ``Optimal load-side control for
  frequency regulation in smart grids,'' \emph{IEEE Transactions on Automatic
  Control}, vol.~62, no.~12, pp. 6294--6309, Dec. 2017.

\bibitem{khalil2002nonlinear}
H.~K. Khalil, \emph{Nonlinear Systems}, 3rd~ed.\hskip 1em plus 0.5em minus
  0.4em\relax Prentice Hall, 2002.

\bibitem{pm2018tcns}
R.~Pates and E.~Mallada, ``Robust scale free synthesis for frequency regulation
  in power systems,'' \emph{IEEE Transactions on Control of Network Systems},
  2019.

\bibitem{oakridge2013}
\BIBentryALTinterwordspacing
G.~Kou, S.~W. Hadley, P.~Markham, and Y.~Liu, ``Developing generic dynamic
  models for the 2030 eastern interconnection grid,'' Oak Ridge National
  Laboratory, Tech. Rep., Dec. 2013. [Online]. Available:
  \url{http://www.osti.gov/scitech/}
\BIBentrySTDinterwordspacing

\bibitem{Horn2012MA}
R.~A. Horn and C.~R. Johnson, \emph{Matrix Analysis}, 2nd~ed.\hskip 1em plus
  0.5em minus 0.4em\relax Cambridge University Press, 2012.

\bibitem{Weigandt1994}
T.~C. {Weigandt}, B.~{ Kim}, and P.~R. {Gray}, ``Analysis of timing jitter in
  {CMOS} ring oscillators,'' in \emph{Proc. of IEEE International Symposium on
  Circuits and Systems}, May 1994, pp. 27--30.

\bibitem{d1973tran}
F.~P. deMello, R.~J. Mills, and W.~F. B'Rells, ``Automatic generation control
  part {II}---{D}igital control techniques,'' \emph{IEEE Transactions on Power
  Apparatus and Systems}, vol. PAS-92, no.~2, pp. 716--724, Mar. 1973.

\bibitem{iceland}
\BIBentryALTinterwordspacing
U.~of~Edinburgh. Power systems test case archive. [Online]. Available:
  \url{https://www.maths.ed.ac.uk/optenergy/NetworkData/icelandDyn/}
\BIBentrySTDinterwordspacing

\bibitem{Dorfler2013kron}
F.~D\"orfler and F.~Bullo, ``Kron reduction of graphs with applications to
  electrical networks,'' \emph{IEEE Transactions on Circuits and Systems I:
  Regular Papers}, vol.~60, no.~1, pp. 150--163, Jan. 2013.

\end{thebibliography}
